\documentclass[openany, amssymb, psamsfonts]{amsart}
\usepackage{mathrsfs,comment}
\usepackage[usenames,dvipsnames]{color}
\usepackage{tikz-cd}
\usepackage{tabularx}
\usepackage[normalem]{ulem}
\usepackage{url}
\usepackage[all,arc,2cell]{xy}
\UseAllTwocells
\usepackage{enumerate}
\usepackage{amsmath, amssymb, amsfonts, mathtools, bm}
\DeclareMathOperator*{\argmax}{arg\,max}
\usepackage{graphicx, float, subcaption, booktabs, array, multirow}
\usepackage{setspace}
\usepackage{enumitem}
\graphicspath{{paper_ready}}
\usepackage{hyperref}  
\hypersetup{%
  bookmarksnumbered=true,%
  bookmarks=true,%
  colorlinks=true,%
  linkcolor=blue,%
  citecolor=blue,%
  filecolor=blue,%
  menucolor=blue,%
  pagecolor=blue,%
  urlcolor=blue,%
  pdfnewwindow=true,%
  pdfstartview=FitBH}

\def\makeautorefname#1#2{\expandafter\def\csname#1autorefname\endcsname{#2}}
\def\equationautorefname~#1\null{(#1)\null}
\makeautorefname{footnote}{footnote}%
\makeautorefname{item}{item}%
\makeautorefname{figure}{Figure}%
\makeautorefname{table}{Table}%
\makeautorefname{part}{Part}%
\makeautorefname{appendix}{Appendix}%
\makeautorefname{chapter}{Chapter}%
\makeautorefname{section}{Section}%
\makeautorefname{subsection}{Section}%
\makeautorefname{subsubsection}{Section}%
\makeautorefname{theorem}{Theorem}%
\makeautorefname{thm}{Theorem}%
\makeautorefname{cor}{Corollary}%
\makeautorefname{lem}{Lemma}%
\makeautorefname{prop}{Proposition}%
\makeautorefname{pro}{Property}
\makeautorefname{conj}{Conjecture}%
\makeautorefname{defn}{Definition}%
\makeautorefname{notn}{Notation}
\makeautorefname{notns}{Notations}
\makeautorefname{rem}{Remark}%
\makeautorefname{quest}{Question}%
\makeautorefname{exmp}{Example}%
\makeautorefname{ax}{Axiom}%
\makeautorefname{claim}{Claim}%
\makeautorefname{ass}{Assumption}%
\makeautorefname{asss}{Assumptions}%
\makeautorefname{con}{Construction}%
\makeautorefname{prob}{Problem}%
\makeautorefname{warn}{Warning}%
\makeautorefname{obs}{Observation}%
\makeautorefname{conv}{Convention}%

\newtheorem{thm}{Theorem}[section]

\newtheorem{prop}{Proposition}[section]

\theoremstyle{definition}
\newtheorem{defn}{Definition}[section]

\makeatletter
\let\c@obs=\c@thm
\let\c@cor=\c@thm
\let\c@prop=\c@thm
\let\c@lem=\c@thm
\let\c@prob=\c@thm
\let\c@con=\c@thm
\let\c@conj=\c@thm
\let\c@defn=\c@thm
\let\c@notn=\c@thm
\let\c@notns=\c@thm
\let\c@exmp=\c@thm
\let\c@ax=\c@thm
\let\c@pro=\c@thm
\let\c@ass=\c@thm
\let\c@warn=\c@thm
\let\c@rem=\c@thm
\let\c@sch=\c@thm
\let\c@equation\c@thm
\numberwithin{equation}{section}
\makeatother

\title[Topological Compression for Forecasting in Prediction Markets]
{Topological Compression of Prediction Market Belief States for Event Forecasting}

\author[A. Jerfy]{Aadit Jerfy}
\address{University of Chicago Department of Mathematics}
\date{July 2026}

\begin{document}

\vspace*{-1cm}

\begin{abstract} 

Topological Data Analysis (TDA) and Persistent Homology (PH) have been utilized in numerous fields$-$from protein modeling to financial analysis$-$to understand the underlying structure of data. In this paper, we incorporate TDA and PH into forecasting models over a prediction market universe to determine whether the belief state of such markets forms an underlying structure whose shape can provide predictive insight. We conclude that single-parameter PH is able to improve predictive ability and calibrated accuracy of linear forecasting models, while two-parameter multipersistence produces similar benefits for nonlinear models. 

\end{abstract}

\maketitle

\tableofcontents

\section{Introduction}
The recent exponential increase in volume on prediction markets has created a natural setting to study the forecasting ability of dispersed beliefs. In such markets, contracts pay out according to the resolution of a future event, with the price of a YES contract commonly interpreted as the market-implied probability of an event. Hence, prediction markets transform opinions and information into a continuously updating probability estimate. 

Under this view, prediction markets are both individual forecasting tools and components of a larger belief system comprised of markets sharing underlying factors that shift price. The motivation of this paper is to question whether that broader belief environment can be compressed in a way that improves forecasting. The simplest forecast for a binary market $B$ is its own implied probability, but if related markets share underlying drivers then there might be additional information in the prices of other markets relating to the resolution of $B$ that is not reflected in its own price. Prior work has studied prediction market accuracy and market-implied probabilities as forecasting instruments \cite{CFTC2026,PolymarketAccuracy}, while separate papers have applied PCA and topological data analysis to financial and high-dimensional data \cite{TradingInterviewPCA,PunXiaLee2018,Blumberg2014}. This paper connects these ideas by applying linear and topological compression directly to prediction market belief states.

The first compression method studied is Principal Component Analysis (PCA) \cite{Pearson1901}, which provides a linear benchmark and asks whether the family-level belief state contains an underlying linear structure relevant to forecasting. 

The natural next step is to ask whether the relevant structure is geometric or purely linear. Topology allows us to describe structure that is not tied to any particular coordinate system, focusing on relationships within data including connectedness, loops, and higher-order organization. Its origins can be traced back to Euler's famous solution to the \textit{Königsberg Bridge Problem}, and since then has become a broad mathematical framework for the study of topological invariants (Definition~\ref{def:topological-invariant}).

The second compression method we utilize is persistent homology (PH), which is designed to detect geometric structure in data across increasing scales. We utilize PH on top of linear PCA models to determine whether topological features of correlated market neighborhoods contain additional forecasting information not captured by initial linear compression.

We also utilize two-dimensional persistent homology coupled with nonlinear benchmarks including XGBoost and Random Forest, which resulted in the strongest overall forecasting models. 

Section 2 introduces relevant background on prediction markets, PCA, decision tree-based learning methods, and persistent homology. Section 3 explores the data pipeline, active-set belief-state construction, PCA benchmark, local Vietoris-Rips topology, multipersistence inspired features, and forecasting models. Section 4 covers empirical results, model comparisons, calibration, regime analysis, and robustness. Section 5 discusses interpretation of these results and limitations, and section 6 concludes with a summary of the implications of the observations on the structure of prediction market belief dynamics.

\section{Background}

\subsection{Prediction Markets}
Since the 2024 presidential election, there has been an influx of activity on platforms such as Kalshi and Polymarket. They support \textit{prediction markets}: markets on which the public is able to trade \textit{event contracts}\footnote{An event contract is a contract whose payoff depends on the resolution of a future event. In the binary case considered here, the contract resolves according to whether the event occurs or does not occur.} -- binary YES / NO contracts that derive their value from the outcome (resolution) of an event \cite{CFTC2026}. For each binary market \(i\) and timestamp \(t\), let
\[
q_{i,t}^{YES}\in[0,1]
\]
denote the price, in dollars, of one YES contract. We define the market-implied probability of event occurrence by
\[
p_{i,t}:=q_{i,t}^{YES}.
\]
This value is interpreted as the market-implied probability that the event occurs, and gives the natural baseline forecast of the market. For example, the market ``Argentina wins the 2026 FIFA World Cup'', with a YES contract priced at $\$0.31$, implies a $31\%$ chance that Argentina wins the tournament, and by extension, a $69\%$ probability that they do not.

Such markets allow for forecasting and hedging of future events, which is why they are viewed as aggregators of truth in the industry \cite{Desarapu2026}. For instance, a trader exposed to Bitcoin (BTC) price declines could use a prediction market tied to a future BTC threshold as a partial hedge: if BTC falls below the threshold, losses on the underlying exposure may be offset by gains from the prediction market position \cite{Bhaskara2026}. YES contract prices are influenced by order flow, with net buying pressure on YES contracts increasing their price, while net selling pressure (buying pressure on the opposing side) reduces it. If a market has sufficient volume, its YES price can be taken as a reflection of public belief of an event occurring, which has been shown in specific cases to be more accurate than even professional estimation \cite{PolymarketAccuracy}.

However, aggregation of public belief is inherently an imperfect predictor due to biases such as herding (crowd-following), overreaction to recent news, partisan or preference-driven trading, and limited attention across related contracts. This issue, combined with the potential for \textit{thin order books}\footnote{A thin order book has relatively few outstanding buy and sell orders, so a small number of trades can move prices substantially.}, \textit{stale prices}\footnote{A stale price is a price that remains unchanged despite new public information that should plausibly affect the probability of the event.}, and arbitrage constraints\footnote{Arbitrage refers to the process of exploiting inconsistent prices across related assets or markets to earn a low-risk profit; arbitrage constraints are frictions that prevent traders from fully correcting such inconsistencies.}, results in even strong market price baselines not being fully efficient.

Limitations of prediction markets include \textit{liquidity fragmentation}\footnote{Liquidity fragmentation is the division of available buyers, sellers, and capital across many similar or related contracts.} across contract thresholds, high fees, wide spreads, manipulation by \textit{whales}\footnote{A whale is a large trader whose orders are large enough to materially move market prices.} and arbitrageurs, market resolution ambiguity, and necessarily binary contracts.
\subsubsection{Prediction Market Belief States}
\begin{defn}
For a selected universe of $n$ prediction markets, the \textit{prediction market belief state} at timestamp $t$ is the vector
\[
p_t=(p_{1,t},p_{2,t},\ldots,p_{n,t})\in[0,1]^n,
\]
where $p_{i,t}$ is the YES-implied probability of market $i$ at time $t$.
\end{defn}
A prediction market belief state is a time-varying aggregation of market information, representing the probability distribution of event occurrence implied by incoming trades.

Related prediction markets are likely to exhibit statistical dependence. If a platform contains separate markets with similar underlying events, then contract prices are driven by the same information shocks. For example, Bitcoin and Ethereum (ETH) price threshold markets, meaning contracts on whether BTC or ETH will trade above or below a specified price by a given date, are both impacted by cryptocurrency regulation news. Thus, a shock can affect many coordinates of $p_t$ simultaneously, creating dependence among the probabilities in the selected market universe. This structure is complicated through consideration of adjacent markets, specifically crypto-adjacent in the case of this paper since we focus on a universe of markets centered around crypto prices. Macro and Fed rates markets, composed of contracts on macroeconomic outcomes such as inflation, interest rates, or Federal Reserve policy decisions, impact crypto markets through expectations surrounding liquidity and dollar strength, while Exchange-Traded Fund (ETF) and crypto policy markets similarly influence crypto markets with expectations about institutional adoption and regulation. Contracts on future actions or outcomes involving the company MicroStrategy are useful in this context as their outcomes are tied to corporate behavior, BTC exposure, and general market sentiment toward Bitcoin as a treasury asset.

A market universe should be viewed as a set of related contracts whose probability processes are statistically dependent due to exposure to shared latent drivers. Thus, our object of interest becomes the evolution of the universe belief state over timestamps $t$.

\subsection{Compression Overview}
At a fixed timestamp $t$, the belief state $p_t$ may contain many coordinates, one for each market in the selected universe. If related markets are driven by common latent factors, then the number of forces moving these probabilities may be much smaller than the number of observed markets. In that case, the observed belief state is high-dimensional in its raw representation, but lower-dimensional in the structure that actually drives its movement.

This motivates compression, the purpose of which is to determine whether a lower-dimensional representation preserves information relevant to forecasting final market resolution. To account for change in the set of markets over time, we utilize an active-set framework where the relevant belief state at timestamp $t$ is constructed exclusively from markets that have opened and have not yet resolved at time $t$. 

Two primary types of compression are explored in this paper: linear and nonlinear. Linear compression is represented by Principal Component Analysis, which captures dominant directions of co-movement, while nonlinear compression, represented by XGBoost \cite{ChenGuestrin2016} and Random Forest \cite{Breiman2001}, allows for the belief state to lie on a curved or more complex surface as opposed to a linear one. Topological summaries consider the shape of the local market neighborhood, searching for connectedness, loops, recurrence (repeated returns to similar belief-state configurations), and clustering structure. In forecasting, compression is only useful if it successfully preserves information relevant to event resolution, which is evaluated with proper scoring metrics including Brier Score \cite{Brier1950} and log loss, both of which compare predicted probabilities to realized binary outcomes. A compression method preserves forecasting information if a model utilizing the compressed state produces better calibrated probability forecasts of event resolution than market price, or if it predicts resolution at higher accuracy.

\subsection{Principal Component Analysis}
Principal Component Analysis (PCA) is a technique utilized for dimensionality reduction of data sets.

PCA is most simply put as a rotation and compression of related variables. Consider Figure \ref{fig:PCA_3D-2D}, in which each point starts off as a three-dimensional market belief state, along axes BTC-family mean implied probability, ETH-family mean implied probability, and Macro rates mean implied probability. These variables move together, which allows PCA to find the two-dimensional plane (defined by PC1 and PC2) preserving the largest amount of variance in the initial point cloud. PC1 and PC2 do not have direct real-world translations, but retain most of the original structure.

\begin{figure}[H]
    \centering
    \includegraphics[width=0.75\linewidth]{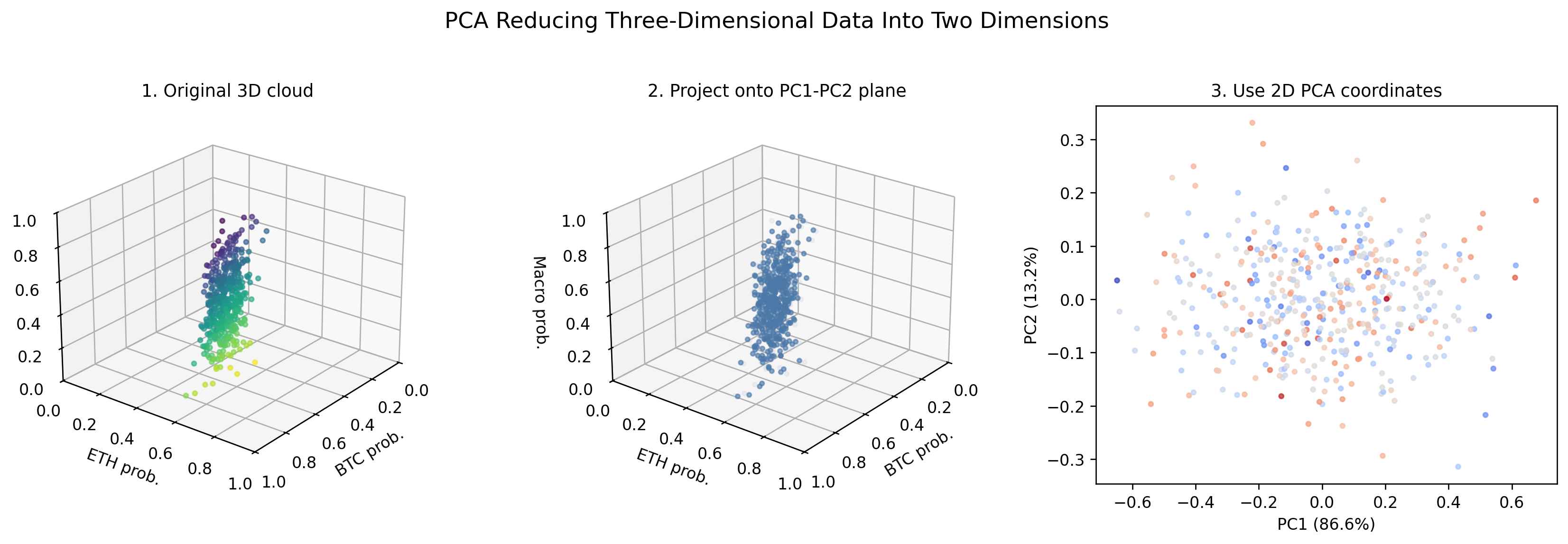}
    \caption{PCA reduction of a three-dimensional market belief state into two principal components}
    \label{fig:PCA_3D-2D}
\end{figure}

\subsubsection{Standardization}
Because PCA compares variables, if one is measured on a far larger scale than another, that variable might be taken as more important simply due to a larger numerical range. Hence, PCA is always preceded by standardization for each feature $x_j$:
\[
\tilde{x}_{t,j}=\frac{x_{t,j}-\overline{x}^{\,train}_j}{s^{train}_j}.
\]
Here, $\overline{x}^{\,train}_j$ is the training-period mean of feature $j$ and $s^{train}_j$ is the training-period standard deviation. Within the context of prediction market belief states, this is essential because active counts, probability means, medians, and dispersions exist on differing scales.

\subsubsection{Covariance Matrix}
In this section, let
\[
X\in \mathbb{R}^{T\times d}
\]
denote the centered and standardized training matrix, where rows correspond to timestamps and columns correspond to $d$ family-state features. PCA begins by forming the covariance matrix
\[
C=\frac{1}{T-1}X^{\mathsf T}X\in\mathbb{R}^{d\times d}.
\]
The $(j,k)$ entry $C_{jk}$ measures the co-movement between feature $j$ and feature $k$ across the training sample. In particular, $C_{jk}>0$ indicates that features $j$ and $k$ tend to rise and fall together, while $C_{jk}<0$ indicates that they tend to move inversely. This is the natural object to examine because PCA looks for directions of joint variation among the features, and $C$ is precisely the matrix encoding that pairwise joint variation \cite{Jaadi2025}.

Linearity is enough at this stage because PCA is intended to serve as the linear benchmark. If the major forecasting-relevant co-movements are approximately linear, they should be visible through the covariance structure of \(X\).

\subsubsection{Principal Components}
The matrix \(C\) is symmetric and positive semidefinite, so the spectral theorem guarantees that its eigenvectors can be chosen to be orthogonal and its eigenvalues are real and nonnegative. These orthogonal eigenvectors are the principal components, and are used as new, uncorrelated axes ordered by the amount of variance they capture \cite{TradingInterviewPCA}. The corresponding eigenvalues measure how much variance is captured along each principal direction.

The eigenvalues can be ordered as
\[
\lambda_1 \geq \lambda_2 \geq \cdots \geq \lambda_d \geq 0.
\]
For $v_k$ the $k$-th eigenvector under this ordering and $\lambda_k$ its corresponding eigenvalue, PCA solves
\[
C v_k = \lambda_k v_k.
\]
We define $PC_i$ to be the $i$-th principal component, so that $PC_i$ captures at least as much variance as $PC_{i+1}$:
\[
PC_1 = \argmax_{\|PC\|=1} \operatorname{Var}(XPC),
\]
\[
PC_k = \argmax_{\substack{\|PC\|=1 \\ PC \perp PC_1,\ldots,PC_{k-1}}} \operatorname{Var}(XPC).
\]
This allows PCA to rank the dimensions of the belief state by informational importance \cite{TradingInterviewPCA}. Each $PC_i$ gives a new axis revealing the greatest spread in data.

Once eigenvectors are ranked, PCA reduces dimensionality by discarding non-leading components. For any component $k$, its explained variance is
\[
EV_k = \frac{\lambda_k}{\sum_{j=1}^{d}\lambda_j},
\]
while the cumulative explained variance after keeping $K$ components is
\[
CEV_K = \frac{\sum_{k=1}^{K}\lambda_k}{\sum_{j=1}^{d}\lambda_j}.
\]
Thus, if the first few components represent the majority of variance, later components can be neglected without losing much of the original structure. The threshold for information one is willing to lose should be optimized for the market set being examined.

\begin{figure}[H]
    \centering
    \includegraphics[width=0.5\linewidth]{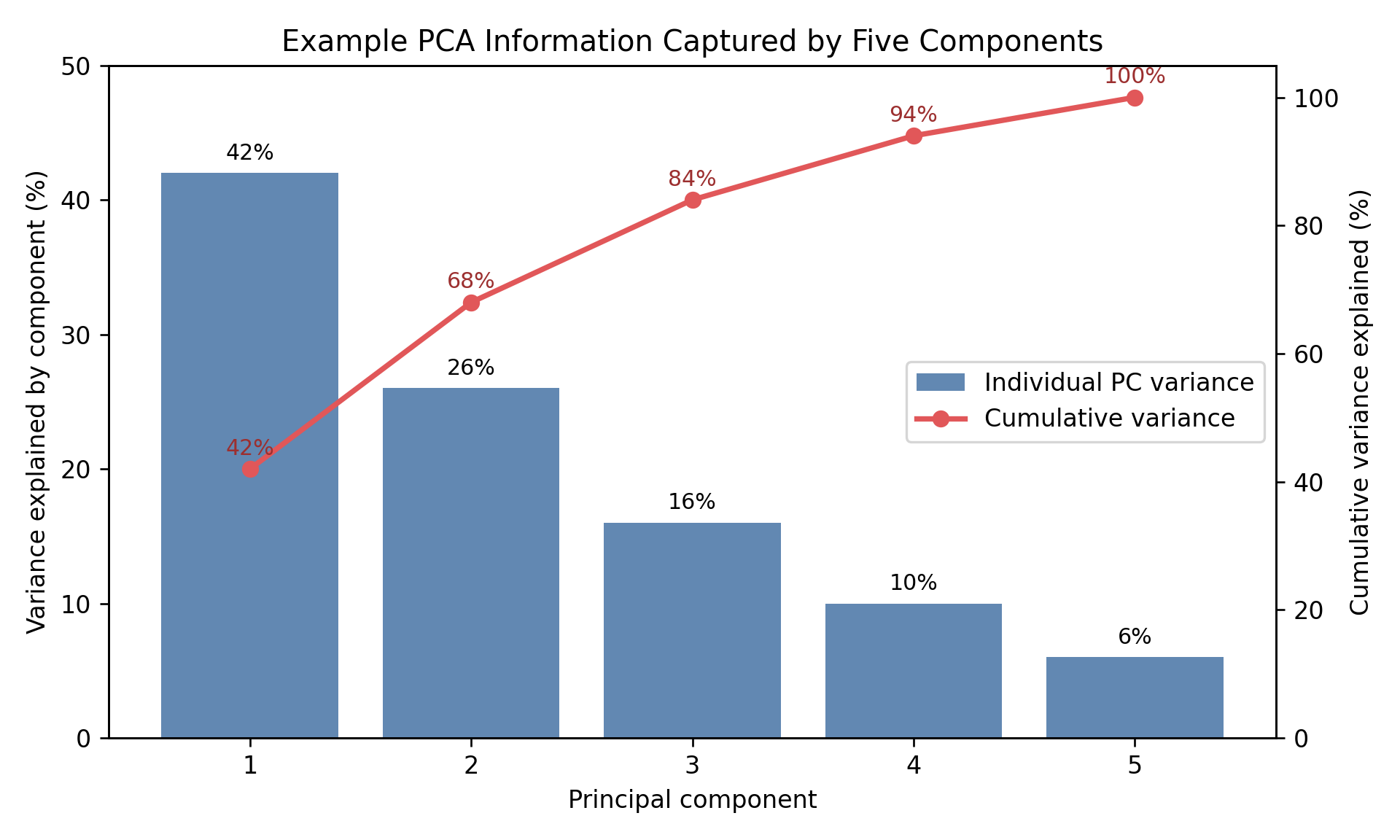}
    \caption{Variance captured from 5-dimensional data by each of 5 principal components. In this example, $84\%$ of the variance is captured within 3 principal components, leaving $16\%$ represented by $PC_4$ and $PC_5$.}
    \label{fig:background-1}
\end{figure}

\subsubsection{Projection}
The family-state belief vector used later in the paper is
\[
F_t = (c_{f,t}, \mu_{f,t}, m_{f,t}, s_{f,t})_{f \in \mathcal F},
\]
where \(c_{f,t}\) is active count, \(\mu_{f,t}\) is mean implied probability, \(m_{f,t}\) is median implied probability, and \(s_{f,t}\) is probability dispersion for family \(f\) at time \(t\). The precise construction is given in Section 3.2.

Since \(F_t\in\mathbb{R}^d\), its standardized version satisfies \(\overline{F}_t\in\mathbb{R}^d\). For \(K\leq d\), let
\[
V_K=(v_1,\ldots,v_K)\in\mathbb{R}^{d\times K},
\]
where each \(v_k\in\mathbb{R}^d\) is a principal direction. We define
\[
z_t^{PCA}=V_K^{\mathsf{T}}\overline{F}_t\in\mathbb{R}^K
\]
to be the compressed linear representation of the global family belief state at time \(t\). In our context, \(z_t^{PCA}\) functions as a variable we utilize to improve forecasts of market resolution. 

In practice, PCA is often computed through \textit{Singular Value Decomposition} (SVD) \cite{Beltrami1873,Jordan1874} rather than the explicit formation and diagonalization of the covariance matrix. The centered and standardized data matrix $X$ is decomposed directly as
\[
X=U S V^{\mathsf T},
\]
where $U$ is a unitary matrix, $V$ is an orthogonal matrix, and $S$ is a diagonal singular-value matrix. The diagonal entries $S_{jj}$ are the singular values of $X$, while off-diagonal entries are zero. The columns of $V$ give the principal directions and the singular values in $S$ determine the variance captured by each direction. Using SVD is computationally more stable, especially for larger or highly correlated datasets, but it is not a different mathematical object from the covariance-based PCA described above. Since
\[
X^{\mathsf{T}}X = V S^2 V^{\mathsf{T}},
\]
the right singular vectors of $X$ are exactly the eigenvectors of the covariance matrix up to the usual scaling factor. Our implementation uses scikit-learn's SVD-based PCA computation \cite{ScikitLearnPCA}.

\subsubsection{PCA-Enhanced Forecasting Model}
For $A_t\in\mathbb{Z}_{\geq 0}$ the active market count at timestamp $t$, and $\Lambda(x)=\frac{1}{1+e^{-x}}$, we formalize our PCA-enhanced forecasting model as
\[
\hat{p}_{i,t}
=
\Lambda(\beta_0+\beta_1p_{i,t}+\beta_2A_t+\gamma^{\mathsf{T}}z_t^{PCA}).
\]
Here, \(\beta_0\) is the intercept representing the model's baseline tendency toward YES resolution before accounting for the observed features; \(\beta_1\) measures the model's reliance on the market-implied probability. A large positive value of \(\beta_1\) means that higher market prices strongly increase the model's predicted probability of YES resolution. \(\beta_2\) measures how the total number of active markets at time \(t\) shifts the forecast after accounting for the individual market price and PCA state. The vector $\gamma$ measures the contribution of the broader compressed belief state; its entries describe how the learned PCA directions adjust the forecast after the individual market price is already included.

Through this model, we pose two key questions:
\begin{enumerate}
    \item Is there a lower-dimensional structure underlying the market belief state?
    \item If so, does this structure contain predictive content?
\end{enumerate}
If PCA improves the Brier Score and log loss of a base logistic regression model that takes $p_{i,t}$ as its primary input, then the answer to both of these questions will be yes, and we will have shown that common linear belief factors contain information not captured entirely by market pricing.

\subsection{XGBoost}
Before introducing persistent homology-enhanced nonlinear models, we briefly review XGBoost as a standard nonlinear benchmark. Prior financial forecasting work has used XGBoost for tasks such as volatility forecasting and stock market prediction \cite{WangWu2025,Jiang2023}, making it a natural comparison model here. We do not assume in advance that nonlinear models must help, but prediction market relationships may include thresholds, interactions, and regime-dependent effects that a linear model cannot represent. XGBoost tests whether the predictive information revealed by PCA and topological features remains useful when the baseline model is allowed to capture such nonlinear interactions. If a topology-enhanced model is only able to improve upon a linear model, it will imply that persistent homology aids in approximating a nonlinear underlying belief-space structure. We first introduce the decision-tree and boosting ideas underlying XGBoost.

\begin{defn}
A \textit{decision tree} is a function \(T:\mathcal X\to\mathbb R\) determined by a finite partition
\[
\mathcal X = \bigcup_{\ell=1}^L R_\ell,
\]
where \(R_\ell \cap R_m=\varnothing\) for \(\ell\neq m\). On each terminal region (subset of the feature space assigned to a leaf node) \(R_\ell\), the tree outputs a constant value \(w_\ell\), so
\[
T(x)=\sum_{\ell=1}^L w_\ell \mathbf{1}\{x\in R_\ell\}.
\]
\end{defn}
In words, a decision tree is a sequentially-structured model used for regression and classification tasks, recursively splitting a feature space into smaller regions.

An example can be visualized in Figure \ref{fig:Example_Decision_Tree} below.

\begin{figure}[H]
    \centering
    \includegraphics[width=0.75\linewidth,trim={0 4cm 0 0},
        clip]{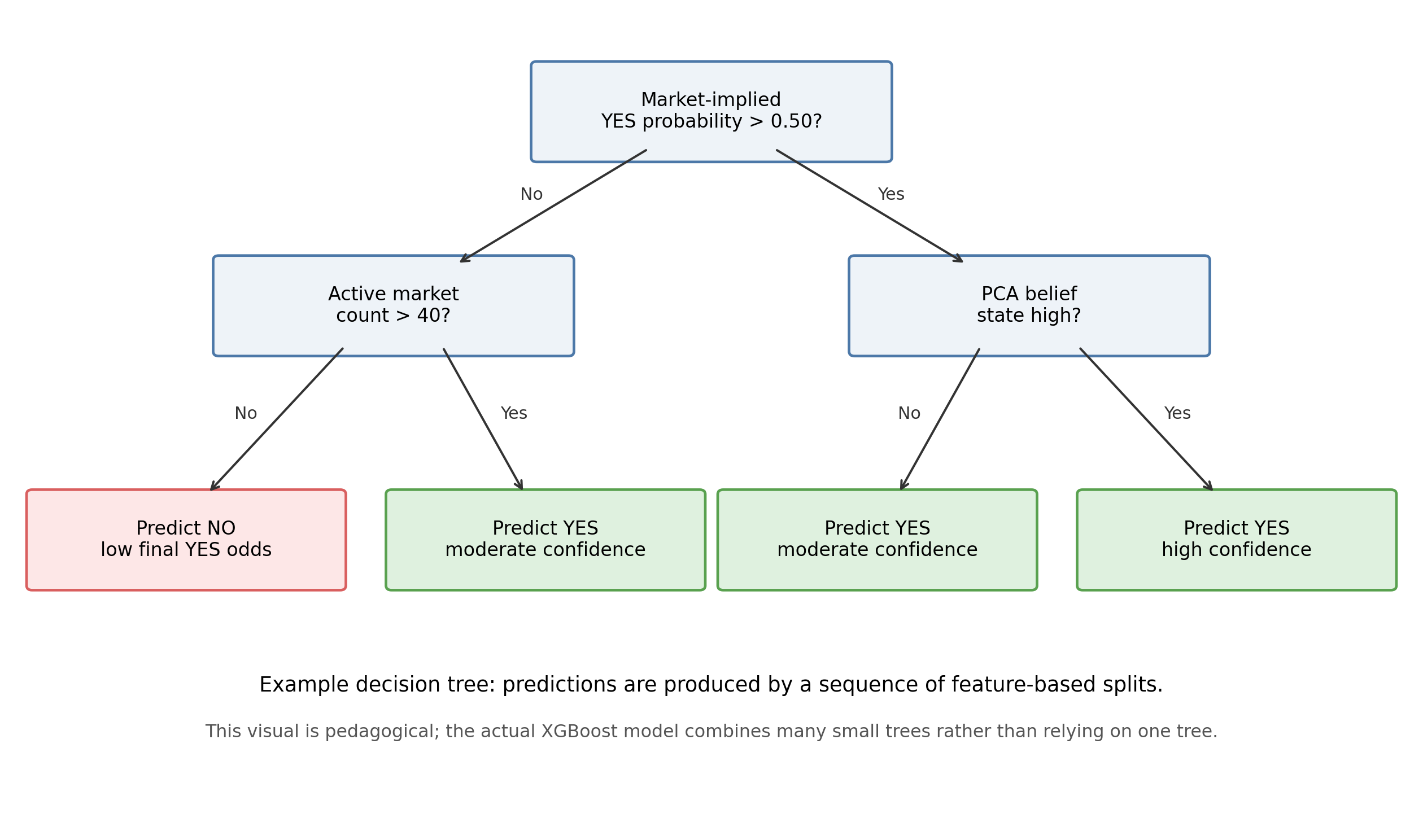}
    \caption{Example Decision Tree}
    \label{fig:Example_Decision_Tree}
\end{figure}

A decision tree in our context might have an initial split on $p_{i,t}$, another on the PCA state, another on active market count, and still more on local neighborhood features addressed in Section 3.5. This allows such trees to treat the same market-implied probability differently depending on the broader market environment.

Some limitations include instability and propensity for overfitting \cite{IBM2026}. This is addressed through a process called boosting, which combines many weak trees into a stronger learning model. Boosting trains weak trees sequentially, with later trees constructed to correct the mistakes of earlier ones \cite{IBM2026}. Trees are fit to reduce the error of the previous model.

\begin{defn}
\textit{Gradient Boosting} is a process combining boosting with utilization of a loss function, such that new trees are fit and added in the direction that minimizes loss. Formally,
\[
\hat{y}_i = \sum_{m=1}^{M} f_m(x_i),
\]
where $f_m$ is a decision tree and $x_i$ is a feature vector for observation $i$.
\end{defn}

Since the process of fitting new trees is additive, we obtain
\[
\hat{y}_i^{\,m} = \hat{y}_i^{\,m-1}+\eta f_m(x_i),
\]
where $\eta\in(0,1]$ is the learning rate, meaning the scalar step size controlling how much each newly added tree changes the current prediction. The ``gradient'' language comes from minimizing a differentiable training loss with respect to the model's current predictions; in our setting, this is justified because the model is trained using smooth probabilistic loss functions such as log loss, even though the final market resolutions are binary. Lower values for $\eta$ help prevent overfitting, as they limit the contribution of any individual tree. Thus, a gradient boosted model is both flexible and regularized, allowing for capture of nonlinear interactions while penalizing overly complex fits to training noise.

Extreme Gradient Boosting (XGBoost) is a machine learning library utilizing gradient-boosted decision trees \cite{IBM2026}. It is essentially a highly efficient implementation of gradient boosting that supports parallel computation, sparsity-aware splitting, and parameter tuning, meaning the selection of modeling parameters such as learning rate, tree depth, number of trees, and regularization strength. In financial data with nonlinear, high-dimensional relationships and instability, XGBoost is a powerful model.

The objective function\footnote{An objective function is the quantity a learning algorithm seeks to minimize during training. It balances how well the model fits the training data against how complex the model is.} optimized by XGBoost can be written generally as
\[
\mathcal{L}
=
\sum_i \ell(y_i,\hat{y}_i)
+
\sum_{m=1}^{M}\Omega(f_m),
\]
where the first term measures prediction error and the second term penalizes model complexity. A common tree-complexity regularization penalty is
\[
\Omega(f)=\gamma T+\frac{1}{2}\lambda\sum_{j=1}^{T}w_j^2,
\]
where $T$ is the number of leaves in the tree, $w_j$ is the score assigned to leaf $j$, and $\lambda$ applies $L_2$-style shrinkage to leaf weights \cite{WangWu2025}. XGBoost is frequently utilized in financial forecasting, with examples ranging from estimation of S\&P 500 volatility index \cite{WangWu2025} to improvement of stock market prediction modeling \cite{Jiang2023}. These examples provide precedent for our analogous setup and utilization of XGBoost as a nonlinear benchmark.

\subsection{Random Forest}
We make use of Random Forest as a second nonlinear benchmark. Random Forest is a nonlinear model built on decision trees, but where XGBoost builds trees sequentially with later trees correcting errors in earlier trees, Random Forest builds many trees independently, combining predictions by calculating their mean \cite{IBMRandomForest2026}. 

Mathematically, for trees $T_1,\ldots,T_B$ and feature vector $x_{i,t}$, the Random Forest probability estimate is
\[
\hat{p}_{i,t}^{RF}=\frac{1}{B}\sum_{b=1}^{B}T_b(x_{i,t}),
\]
where $B$ is the total number of trees and each $T_b(x_{i,t})\in[0,1]$ is a scalar probability estimate produced by tree $b$, not a full probability distribution. The core of Random Forest is bootstrap aggregation, also known as bagging \cite{IBMRandomForest2026}. In this process, each tree is trained on a randomly sampled version of the training data, such that individual trees see related, but non-identical datasets. Randomness is also introduced at the feature level, allowing each node split to consider only a random subset of features \cite{IBMRandomForest2026}, which prevents trees from becoming too similar to one another. This protects the strength of the ensemble of trees, which is strongest when individual components make imperfectly correlated errors.

This allows Random Forest to detect nonlinear threshold effects, local interactions, and conditional relationships that might be invisible to a logistic regression model. For example, a Random Forest can learn that a PCA coordinate matters only when the market-implied probability is already high, or that a topology feature matters only within a particular range of active market counts. There are also limitations to Random Forest. To begin with, although individual decision trees are relatively interpretable, a large forest is harder to interpret as the final prediction is the average of many individual tree predictions. Additionally, fitting many trees requires storing many split rules and evaluating each tree at prediction time, making Random Forest more computationally and memory intensive than simpler linear models. In this paper, Random Forest is used as a nonlinear robustness benchmark. 
\subsection{Topology Fundamentals}
Before diving into the specific methodology undertaken in this paper, it is important to obtain an understanding of certain basic topological notions.

\begin{defn}\label{def:topological-invariant}
A \textit{topological invariant} is a quantity or object associated to a space that remains unchanged under continuous deformation. More formally, if two spaces are homeomorphic, then a topological invariant assigns the same value or object to both spaces.
\end{defn}

The topological invariants of particular interest to us are derived from homology (Section~\ref{subsec:homology}). We are especially concerned with the number of connected components and loops that persist across a filtration.

\subsubsection{Simplicial complexes}
To make sense of the ``shape'' of a market neighborhood, we first need a discrete combinatorial model of that shape. Simplicial complexes provide this model by allowing vertices, edges, triangles, and higher-dimensional analogues to be treated within a single object.

Suppose $V$ is a finite, nonempty set containing vertices $v$.

\begin{defn}
A \textit{simplicial complex on $V$} is a collection $S$ of nonempty subsets of $V$ such that:
\[
v \in V \implies \{v\}\in S,
\]
and
\[
\tau \in S,\ \sigma \subseteq \tau,\ \sigma\neq \emptyset \implies \sigma \in S.
\]
The elements $\sigma,\tau$ of $S$ are called \textit{simplices}.
\end{defn}

This definition elicits a notion of dimension of $S$:
\[
\dim S = \max\{\dim \sigma:\sigma \in S\},
\]
where $\dim \sigma$ is defined to be one less than the number of vertices in $\sigma$. Hence,
\[
S_i=\{\sigma\in S:\dim \sigma=i\}.
\]
Examples of simplicial complexes include undirected graphs and various geometric shapes.

\begin{figure} [H]
    \centering
    \includegraphics[width=0.5\linewidth]{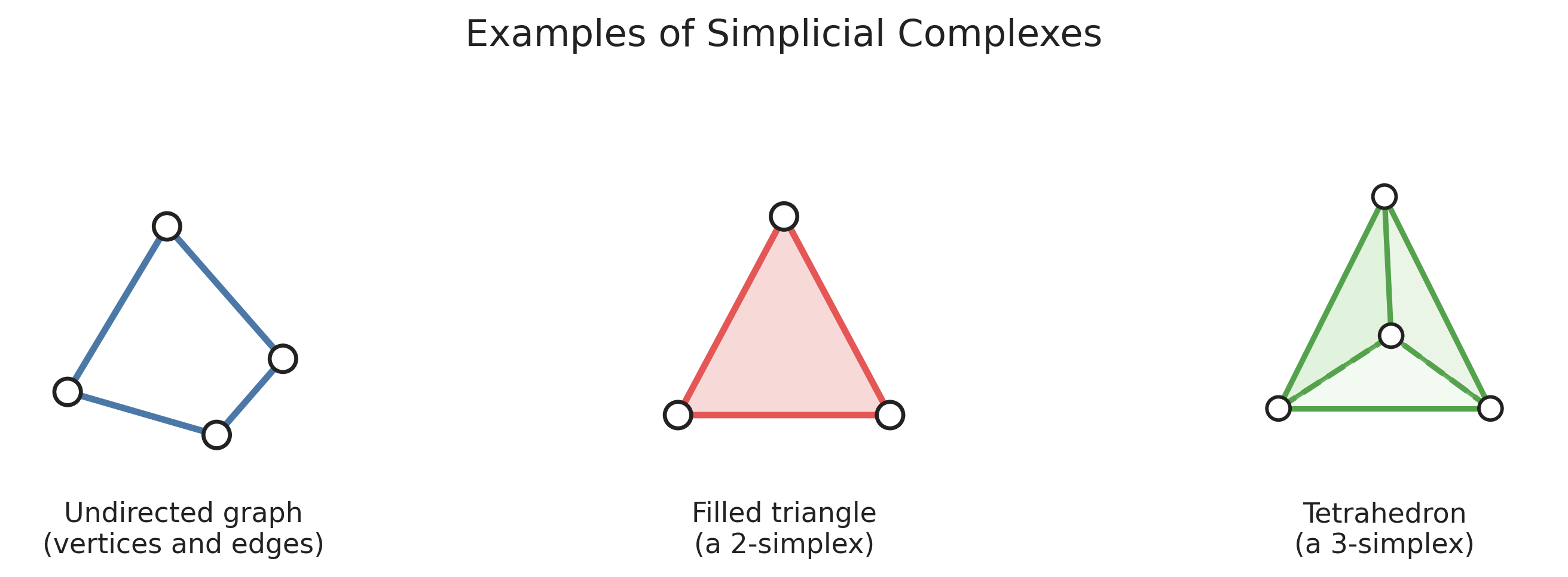}
    \caption{Examples of simplicial complexes}
    \label{fig:background-2}
\end{figure}

\begin{defn}
Let $S$ be a simplicial complex, and take $\sigma,\tau\in S$. The simplex $\sigma$ is a \textit{face of $\tau$}, denoted $\sigma\leq\tau$, if $\sigma\subseteq\tau$. 
\end{defn}

Every vertex in $\sigma$ must also be a vertex in $\tau$ in order for $\sigma$ to be a face of $\tau$. If $\sigma\subsetneq\tau$, then $\sigma$ is a proper face of $\tau$.

\subsubsection{Filtrations}
Persistent homology measures the evolution of topological features as a scale parameter increases. To formalize this evolution, we need the notion of a nested sequence of subcomplexes.

\begin{defn}
Take $S$ to be a simplicial complex. A collection $T\subseteq S$ is a \textit{subcomplex} of $S$ if $T$ is itself a simplicial complex and, for every $\tau\in T$, every face $\sigma\leq\tau$ also satisfies $\sigma\in T$.
\end{defn}

An example of a subcomplex can be seen below, with
\[
T=\{\{a\},\{b\},\{c\},\{ab\},\{bc\},\{ac\}\}
\]
forming a boundary subcomplex of the filled triangle
\[
S=\{\{a\},\{b\},\{c\},\{ab\},\{bc\},\{ac\},\{abc\}\}.
\]

\begin{figure}[H]
    \centering
    \includegraphics[width=0.2\linewidth]{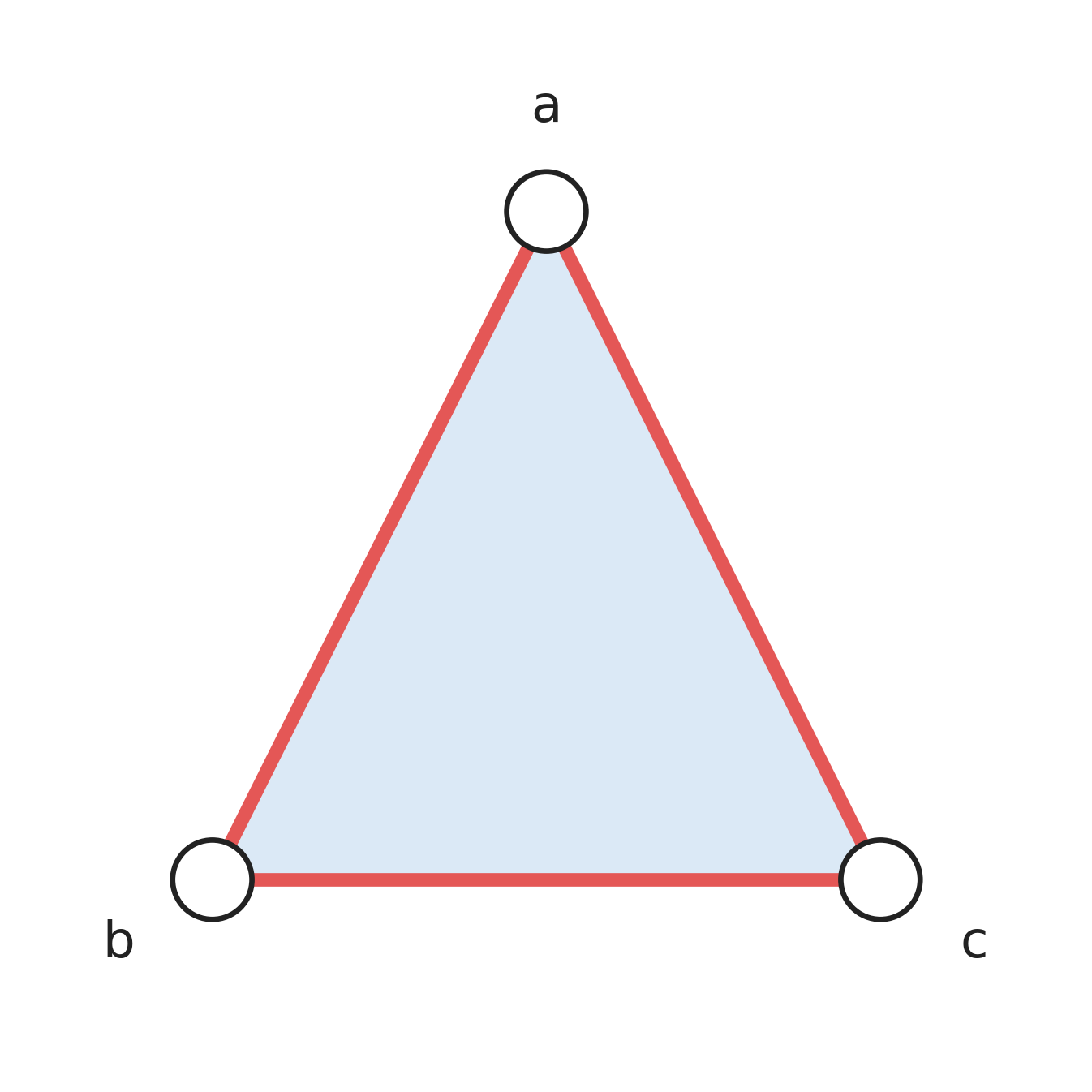}
    \caption{Subcomplex Visualization}
    \label{fig:background-3}
\end{figure}

Now we are in a position to introduce a structure crucial to persistent homology, the filtration.

\begin{defn}
A \textit{filtration} of $S$ indexed by $\mathbb{Z}$ is a nested sequence of subcomplexes
\[
F_1S \subset F_2S \subset \cdots \subset F_nS=S.
\]
A filtration indexed by $\mathbb{R}$ is a family $\{F_rS\}_{r\in\mathbb{R}}$ such that $F_rS\subseteq F_sS$ whenever $r\leq s$.
\end{defn}

\subsubsection{Vietoris-Rips Complexes}
The specific filtration used throughout this paper is the Vietoris-Rips filtration \cite{Vietoris1927}, which is built from pairwise distances between points. A simplex is included whenever all of its vertices are pairwise within distance $\varepsilon$. By contrast, the more geometrically faithful \v{C}ech complex \cite{Cech1932} includes a simplex only when the $\varepsilon$-balls around all of its vertices have a common intersection. Such a complex is harder to compute in high-dimensional settings, so the Vietoris-Rips construction is a more practical choice. 

Consider a point cloud $P\subset \mathbb{R}^n$. There exists a well-defined notion of \textit{thickening} $P$ by any scale $\varepsilon \in \mathbb{R}_+$, which elicits the notation $P^{+\varepsilon}$ defined to be the union of $\varepsilon$-balls about each point $p\in P$. Formally,
\[
P^{+\varepsilon}=\bigcup_{p\in P}B_\varepsilon(p).
\]

\begin{figure} [H]
    \centering
    \includegraphics[width=0.75\linewidth,trim={0 0 0 1.5cm},
        clip]{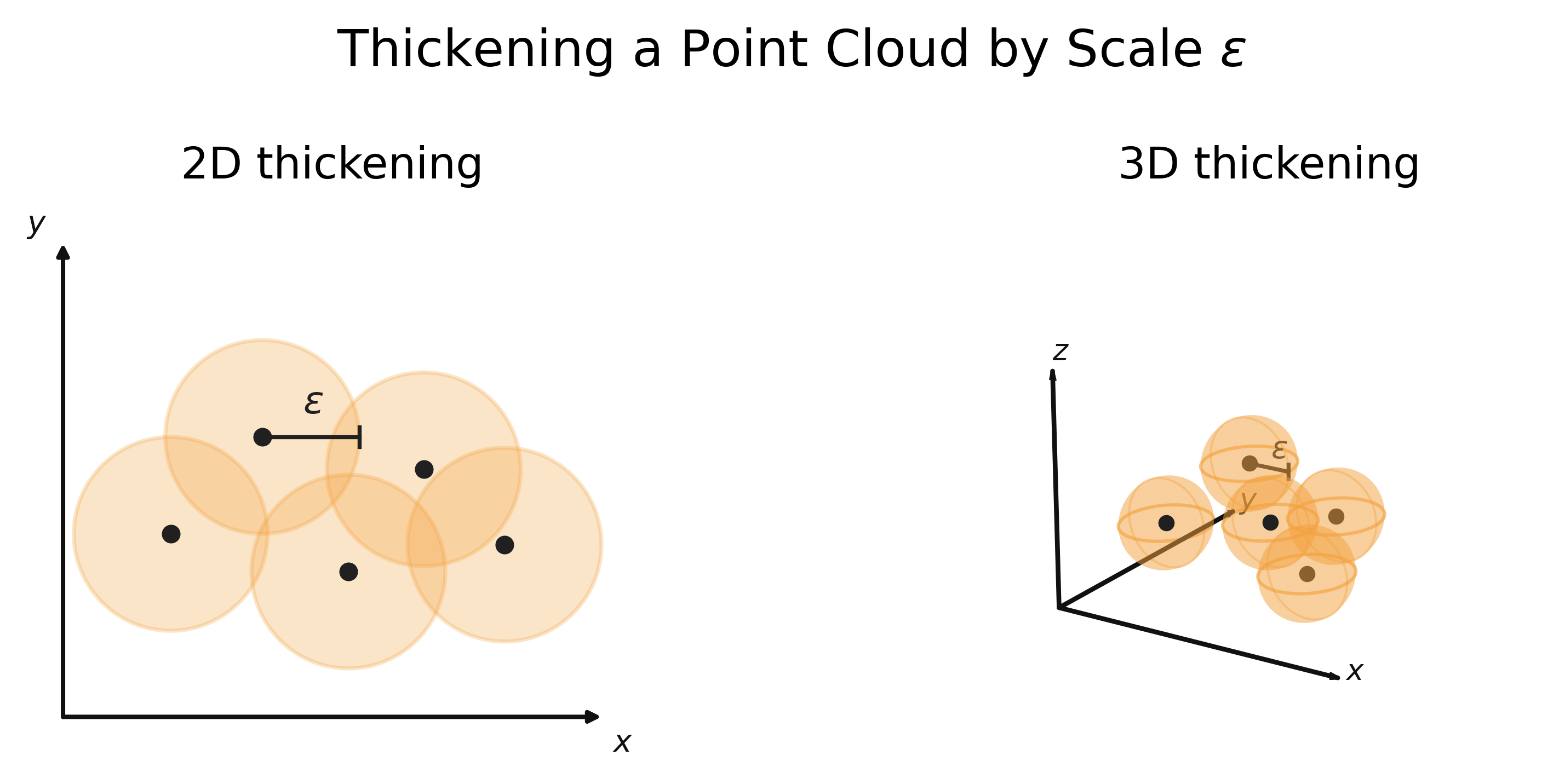}
    \caption{Thickening of point cloud $P$ by $\varepsilon$.}
    \label{fig:background-4}
\end{figure}

\begin{defn}
Let $(P,d)$ be a finite metric space with metric $d$. The \textit{Vietoris-Rips filtration} of $P$ is the collection of simplicial complexes $\{VR_\varepsilon(P)\}_{\varepsilon\in\mathbb{R}_+}$ such that, for any finite subset $\{p_0,\ldots,p_n\}\subset P$,
\[
\{p_0,\ldots,p_n\}\text{ forms an }n\text{-simplex in }VR_\varepsilon(P)
\iff
d(p_i,p_j)\leq \varepsilon\quad \forall\, i,j.
\]
\end{defn}

It follows that $VR_\varepsilon(P)$ is a simplicial complex for all $\varepsilon>0$, with elements $p\in P$ forming vertices of these complexes. Another important result is that for all pairs $\delta\leq \varepsilon$, we have
\[
VR_\delta(P)\subseteq VR_\varepsilon(P).
\]
The formation of such Vietoris-Rips complexes can be thought of as a thickening of the neighborhood around a given point of interest. As $\varepsilon$ increases, additional edges and higher-dimensional simplices are added, forming a nested sequence.

\begin{figure} [H]
    \centering
    \includegraphics[width=0.75\linewidth]{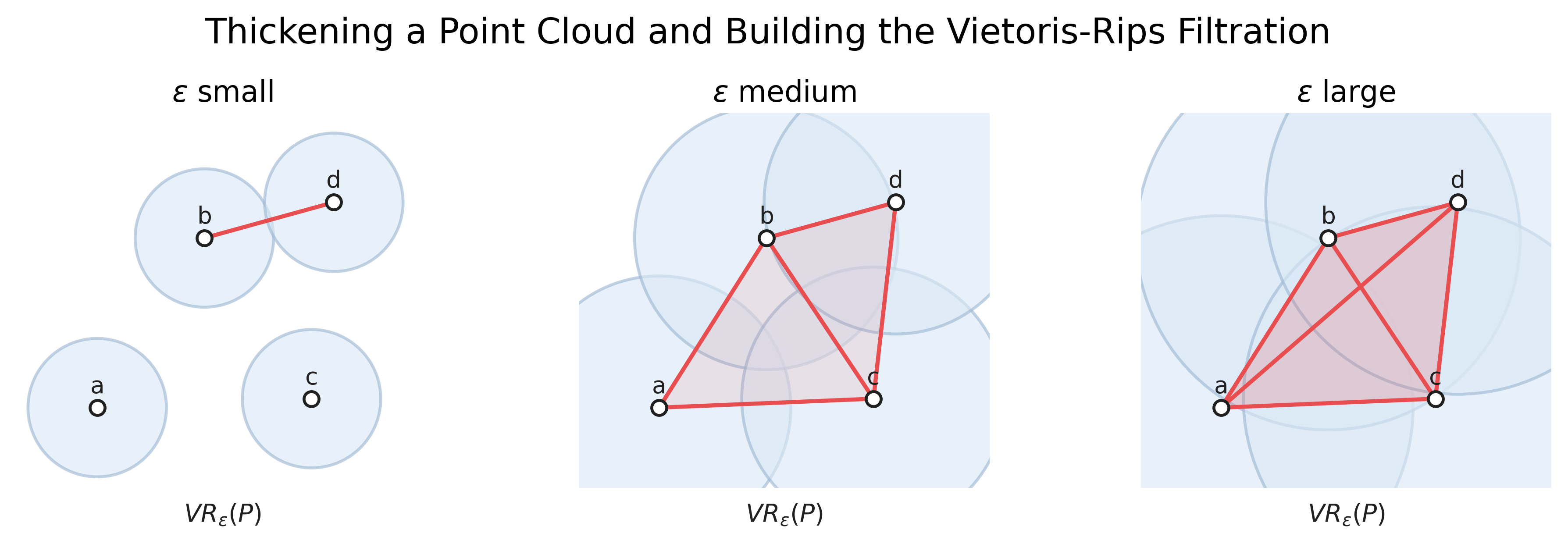}
    \caption{Vietoris-Rips Filtration}
    \label{fig:background-5}
\end{figure}

\subsubsection{Homology}\label{subsec:homology}
In topology, we study shapes and structures within \textit{homotopy equivalence classes}, taking shapes with the same holes and intersections as equivalent. For instance, a hollow triangle is homotopy equivalent to a hollow square because both can be continuously deformed into a circle as each has one connected component and one one-dimensional hole.

Homology is the algebraic tool by which we count connected components (0-dimensional features), loops (1-dimensional features), voids (2-dimensional features), and higher-dimensional analogues. To count these structures algebraically we introduce orientations, boundaries, chain groups, and homology groups.

\begin{defn}
An \textit{orientation} of a simplicial complex $P$ is a choice of ordering for the vertices of each simplex in $P$.
\end{defn}

\begin{defn}
The \textit{$i$-th face} of an oriented $p$-simplex
\[
\sigma=(v_0,v_1,\ldots,v_p)
\]
is the $(p-1)$-simplex
\[
\sigma_{-i}=(v_0,\ldots,v_{i-1},v_{i+1},\ldots,v_p)
\]
obtained by removing the $i$-th vertex.
\end{defn}

Consider the following triangle with vertices $0,1,2$. We have that $\sigma_{-0}=\{1,2\}$, $\sigma_{-1}=\{0,2\}$, and $\sigma_{-2}=\{0,1\}$.

\begin{figure} [H]
    \centering
    \includegraphics[width=0.25\linewidth]{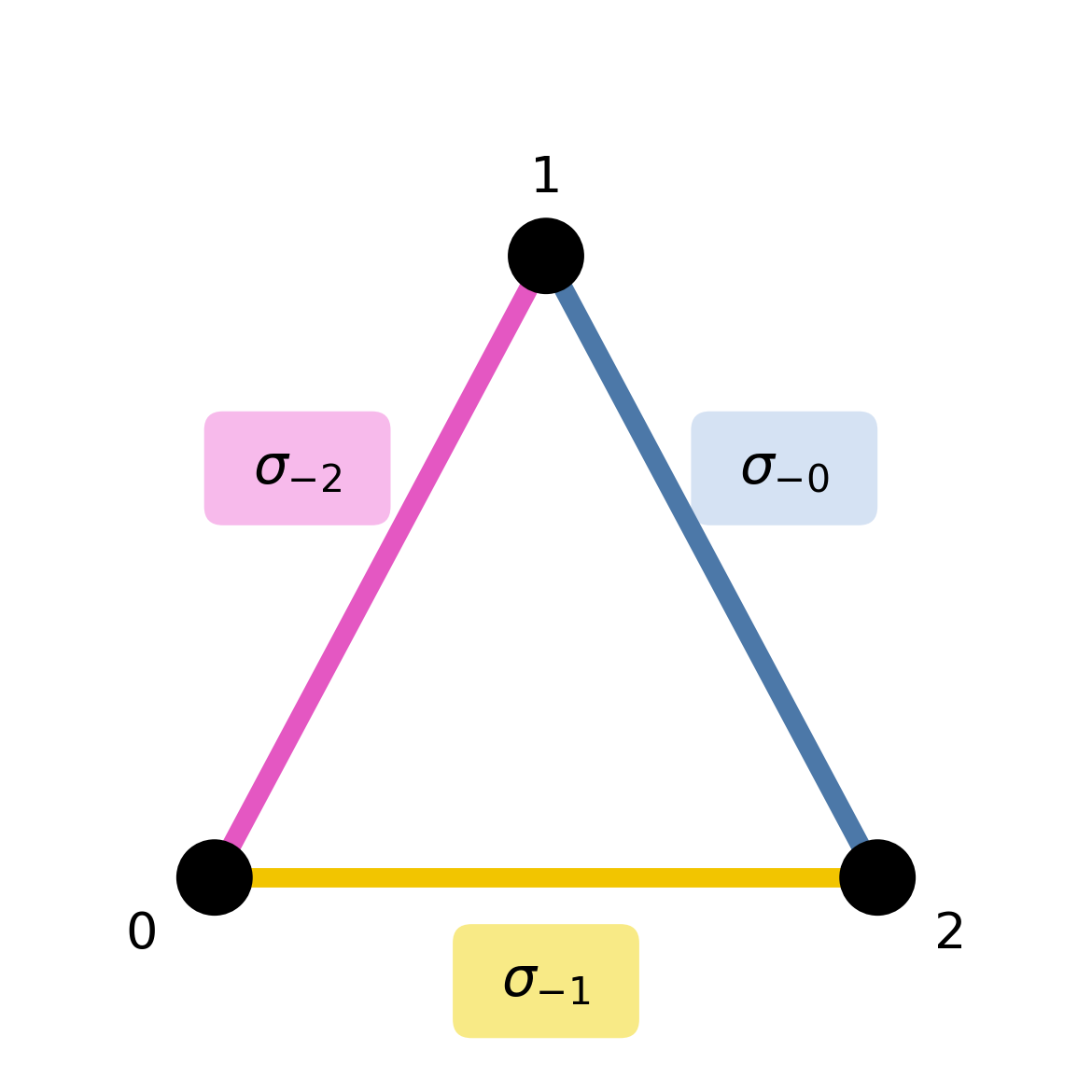}
    \caption{Faces of a triangle simplex}
    \label{fig:background-6}
\end{figure}

\begin{defn}
The \textit{$p$-th chain group} of $P$, denoted $C_p(P)$, is the vector space generated by the $p$-simplices of $P$ as basis elements.
\end{defn}

Thus, each $\gamma\in C_p(P)$ can be uniquely expressed as
\[
\gamma=\sum_{\sigma}\gamma_\sigma\cdot\sigma,
\]
where $\gamma_\sigma\in\mathbb{F}$ and $\sigma$ ranges over the $p$-simplices of $P$.

\begin{defn}
Let $\sigma$ be a $p$-dimensional oriented simplex. The \textit{algebraic boundary} of $\sigma$ is
\[
\delta_p\sigma=\sum_{i=0}^{p}(-1)^i\sigma_{-i}.
\]
\end{defn}

The boundary operator
\[
\delta_p:C_p(P)\rightarrow C_{p-1}(P)
\]
maps a simplex to its boundary.

Conveniently, we have that $\delta_{p-1}\circ\delta_p=0$ for $p>0$, which is the algebraic fact making homology possible.

\begin{prop}
For every \(p>0\),
\[
\delta_{p-1}\circ \delta_p=0.
\]
\end{prop}

\begin{proof}
It suffices to check this on an oriented \(p\)-simplex
\(\sigma=[v_0,\ldots,v_p]\). By definition,
\[
\delta_p\sigma=\sum_{i=0}^{p}(-1)^i
[v_0,\ldots,\widehat{v_i},\ldots,v_p].
\]
Applying \(\delta_{p-1}\), every \((p-2)\)-face obtained by deleting two vertices \(v_i,v_j\) with \(i<j\) appears twice: once by deleting \(v_i\) first and then \(v_j\), and once by deleting \(v_j\) first and then \(v_i\). These two appearances have coefficients
\[
(-1)^i(-1)^{j-1}
\qquad\text{and}\qquad
(-1)^j(-1)^i,
\]
which sum to zero. Hence all terms cancel in pairs, so
\[
\delta_{p-1}\delta_p\sigma=0.
\]
By linearity, \(\delta_{p-1}\circ\delta_p=0\) on \(C_p(P)\).
\end{proof}

Elements of \(\ker \delta_p\) are called \(p\)-cycles, while elements of \(\operatorname{im}\delta_{p+1}\) are called \(p\)-boundaries.

\begin{defn}
The \textit{$p$-th homology group of $P$} is
\[
H_p(P)=\ker \delta_p / \operatorname{im}\delta_{p+1}.
\]
\end{defn}

Hence, homology groups are defined by cycles modulo boundaries. We write
\[
\beta_p(P)=\dim H_p(P),
\]
where $\beta_p$ is called the $p$-th \textit{Betti number} of $P$. A Betti number entirely determines $H_p(P)$ up to isomorphism as a vector space.

\subsubsection{Persistent Homology}
Persistent homology is a technique within topological data analysis revolving around the measurement of the lifespans of various topological structures, namely connected components, loops, voids, and higher-dimensional analogues.

To consider the lifespan of such structures, we must define ideas of \textit{birth} and \textit{death}. The birth of a feature is the first stage of a filtration in which it appears, while the death of a feature is the stage where that feature disappears or merges into another feature. To formalize this, we introduce persistence modules and persistent homology groups.

\begin{defn}
A \textit{persistence module} indexed by \(\mathbb{Z}\) is a collection of vector spaces
\[
\{V_i\}_{i\in\mathbb{Z}}
\]
together with linear maps
\[
a_{i,j}:V_i\to V_j
\qquad \text{for every } i\leq j,
\]
such that
\[
a_{i,i}=\operatorname{id}_{V_i}
\]
and, whenever \(i\leq j\leq k\),
\[
a_{i,k}=a_{j,k}\circ a_{i,j}.
\]
Equivalently, if the module is specified by maps \(a_i:V_i\to V_{i+1}\), then for \(i<j\) we define
\[
a_{i,j}=a_{j-1}\circ a_{j-2}\circ\cdots\circ a_i.
\]
\end{defn}

\begin{defn}
The $j$-th \textit{persistent homology group} of $(V,a)$ is the subspace of $V_j$ given by
\[
PH_{i\rightarrow j}=\operatorname{im}(a_{i\rightarrow j}).
\]
If $F(P)$ is our filtration, analysis of such groups involves computation of homology groups $H_p(F_i(P))$ in order to obtain $p$-dimensional features at stage $i$.
\end{defn}
Now we can say that some $v \in V_i$ is born at $i$ if $v \notin$ im$(a_{i-1})$, while the same $v$ dies at $j$ if $j$ is the smallest number such that $a_{i \rightarrow j}(v) = 0.$ Thus, $PH_{i\rightarrow j}$ can be thought of as the set of p-dimensional features born at or before stage $i$ that survive into stage $j$. 

This is where Betti numbers come into play. Since $\beta_p =$dim$H_p$, the Betti numbers count how many independent $p$-dimensional topological features are present at a given stage of filtration.  $\beta_0$ counts connected components, $\beta_1$ counts independent loops, and $\beta_2$ counts enclosed voids. In this paper, we examine $H_0$ and $H_1$ because the relevant structures in prediction market belief geometry are the formation of connected market clusters and the presence of loop-like relationships among correlated markets. Persistent homology tracks the appearance of these components and loops in addition to how long they survive with the change of filtration scale, allowing us to identify stable and intermittent geometric structure. 
\subsubsection{Persistence Diagrams/Barcodes}
The information produced by persistent homology is most commonly summarized through \textit{persistence diagrams} and \textit{barcodes} \cite{ZomorodianCarlsson2005,CrawleyBoevey2015}.

In the single-parameter case, persistence diagrams and barcodes provide a clean way to summarize topological features through birth-death intervals. The basic object behind this representation is an interval module: a persistence module supported on one interval of the filtration parameter, with one copy of the field on that interval and zero outside it. This interval-based description makes single-parameter persistence especially easy to visualize.

\begin{thm}[Interval decomposition {\cite{ZomorodianCarlsson2005,CrawleyBoevey2015}}]\label{thm:interval-decomposition}
Let $V$ be a persistence module indexed by $\mathbb{Z}$ or $\mathbb{R}$ over a field $\mathbb{F}$, with $\dim_{\mathbb{F}}V_t<\infty$ at every index $t$. Then
\[
V\cong\bigoplus_{\alpha\in A}\mathbb{F}_{I_\alpha},
\]
where each $I_\alpha$ is a nonempty interval. Here, $\mathbb{F}_I$ has value $\mathbb{F}$ at indices in $I$ and zero outside $I$, with identity structure maps between indices in $I$ and zero maps otherwise. The multiset of intervals is unique up to reordering and is the barcode of $V$.
\end{thm}

If a topological feature is born at filtration scale $i$ and dies at scale $j$, then it can naturally be represented as a point $(i, j)$ in $\mathbb{R}^2$. The collection of all such points for a fixed dimension $p$ is called the $p$-dimensional persistence diagram. Points close to the diagonal $i = j$ represent features that die shortly after birth, and are therefore taken as noise or short-lived structure. Points far from the diagonal have long lifespans, implying persistence across numerous filtration scales, and are more likely to reflect stable geometric structure within the data.
\begin{figure}[H]
    \centering
    \includegraphics[width=0.4\linewidth,trim={0 0.6cm 0 0},
        clip]{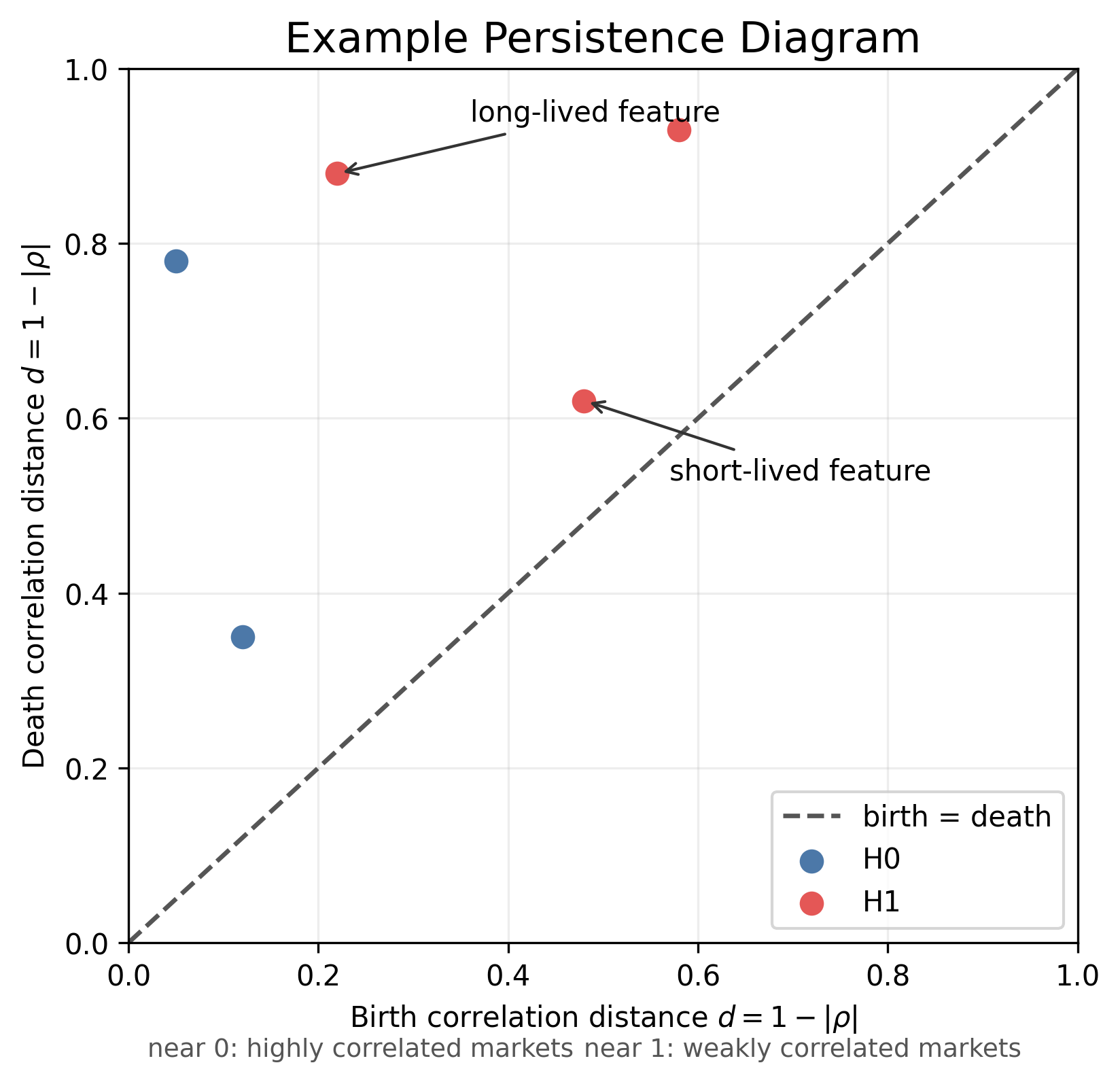}
    \caption{Example Persistence Diagram}
    \label{fig:background-7}
\end{figure}

More often, however, this information is represented as a \textit{barcode}. Points $(i,j)$ in the analogous persistence diagram are instead represented as intervals $[i, j)$, beginning at the scale where the feature appears and ending where it disappears. Longer bars correspond to more persistent features, while shorter bars correspond to features that only briefly appear.

\begin{figure}[!htbp]
    \centering
    \includegraphics[width=0.5\linewidth]{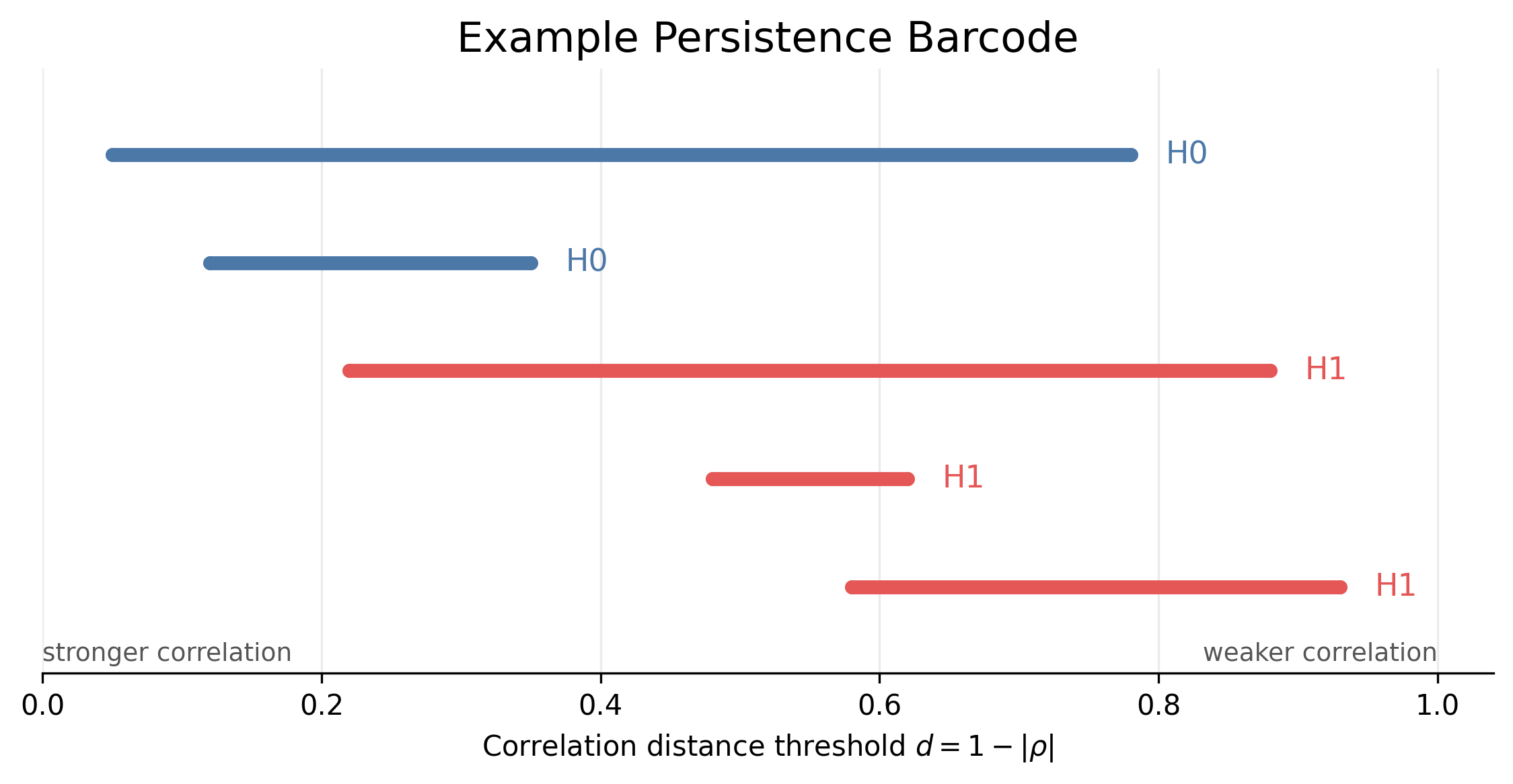}
    \caption{Example persistence barcode}
    \label{fig:persistence_barcode}
\end{figure}
\subsubsection{Bottleneck Distance}
We require a method by which to determine whether two persistence modules are ``close'' to one another. For pointwise finite-dimensional single parameter persistence modules, this is commonly done by first passing to their persistence diagrams or barcodes and then comparing those summaries using bottleneck distance (Definition~\ref{def:bottleneck-distance}).

\begin{defn} Let $\varepsilon \geq 0$. An $\varepsilon$\textit{-matching} between two barcodes $D_A$ and $D_B$ is a bijection
\[
\rho: A_0 \rightarrow B_0
\]
between multi-subsets $A_0 \subseteq D_A$ and $B_0 \subseteq D_B$ such that
\[
\begin{aligned}
&\text{1) every unmatched interval } [a,b] \in (D_A-A_0)\cup(D_B-B_0)\\
&\qquad \text{has length } b-a \leq 2\varepsilon,\\
&\text{2) if } \rho([a_1,a_2])=[b_1,b_2],\\
&\qquad \text{then } |a_1-b_1|\leq \varepsilon
\text{ and } |a_2-b_2|\leq \varepsilon.
\end{aligned}
\] \end{defn}

\begin{defn}\label{def:bottleneck-distance} The \textit{bottleneck distance} between barcodes $D_A$ and $D_B$ is defined by
\[
d_{\mathrm{Bot}}(D_A,D_B)
=
\inf\{\varepsilon \geq 0 :
\text{ there exists an } \varepsilon\text{-matching between } D_A \text{ and } D_B\}.
\] \end{defn}
The bottleneck distance describes the minimum amount you would need to move the worst-matched feature so that two persistence diagrams line up, allowing short-lived unmatched features to be matched to the diagonal. 
\begin{figure}
    \centering
    \includegraphics[width=0.5\linewidth]{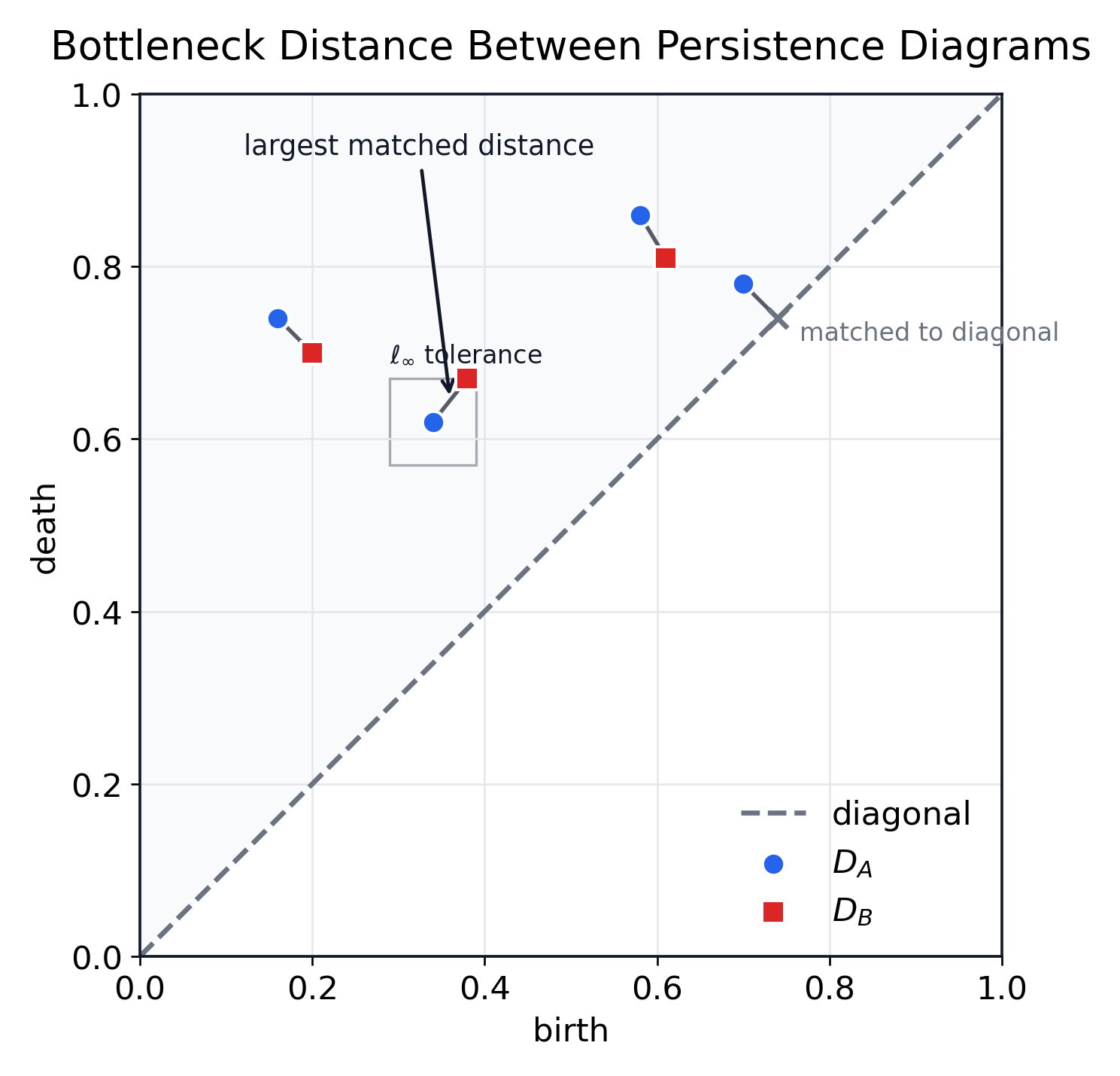}
    \caption{Bottleneck Distance between two Persistence Diagrams}
    \label{fig:background-8}
\end{figure}

\subsubsection{Stability}
\begin{thm}[Stability of persistence diagrams {\cite{CohenSteinerEdelsbrunnerHarer2007}}]\label{thm:persistence-stability}
Let $K$ be a finite simplicial complex and let $f,g:|K|\to\mathbb{R}$ be continuous piecewise-linear functions on its geometric realization. Fix a homological degree $p\geq 0$ and a coefficient field $\mathbb{F}$. Let $D_p(f)$ and $D_p(g)$ be the persistence diagrams of the sublevel-set filtrations
\[
\{x\in |K|:f(x)\leq a\},\qquad
\{x\in |K|:g(x)\leq a\},\qquad a\in\mathbb{R}.
\]
Then
\[
d_{\mathrm{Bot}}\bigl(D_p(f),D_p(g)\bigr)
\leq \|f-g\|_\infty
=\sup_{x\in |K|}|f(x)-g(x)|.
\]
Intervals with infinite death time are matched to one another, with their matching cost determined by the difference in birth times.
\end{thm}

Why are the lifespans of $H_0$ and $H_1$ structures of concern? One reason is that analysis with consideration of lifespan is stable under small data perturbations \cite{CohenSteinerEdelsbrunnerHarer2007}. If a point cloud is slightly shifted or if market probabilities move by a small amount, then the resulting persistence diagram should change only by a similarly small amount. Prediction market data is noisy, so small fluctuations should not produce entirely different topological conclusions. 
\subsubsection{Multipersistence}
Despite the various advantages of persistent homology, in practice it is often the case that single-parameter persistent homology is not sufficient as an analytical tool. A single scale parameter applies the same distance threshold across the entire dataset, which can blur distinctions between dense regions (where small-scale structure is meaningful) and sparse regions (where the same scale may connect points too aggressively) \cite{Blumberg2014}.
Figure~\ref{fig:single_parameter_density_issue} provides a visual representation of this issue.

\begin{figure}[H]
    \centering
    \includegraphics[width=1\linewidth, trim={0 0 0 3cm},
        clip]{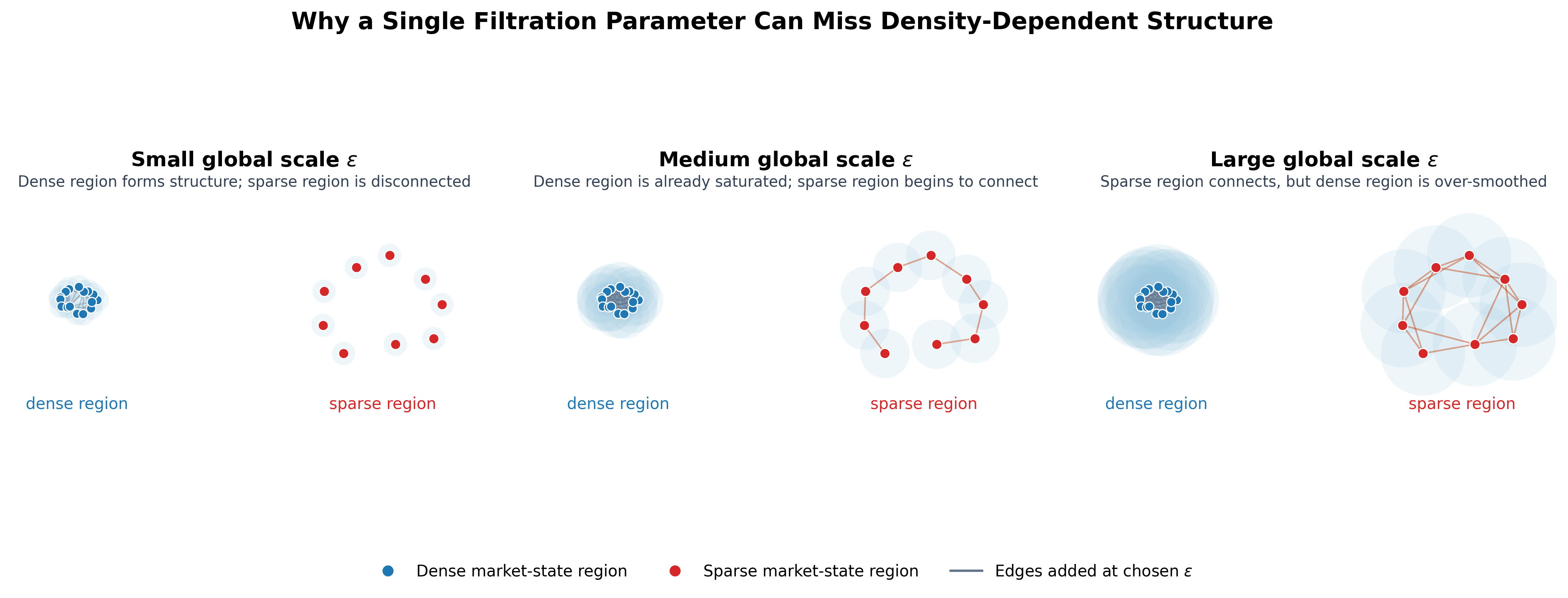}
    \caption{Single-parameter PH in dense and sparse regions}
    \label{fig:single_parameter_density_issue}
\end{figure}  

A natural extension is to include an additional filtration parameter. By allowing scale and density to vary as two separate filtration parameters, one can track which topological features persist not only across distance thresholds but also across different levels of local data density. This produces a more robust analytical tool because features must remain stable under two complementary notions of resolution. To formalize this procedure, we require a few more definitions. 
\begin{defn}
A partially ordered set (\textit{poset}) \(P\) is a set equipped with a relation \(\leq\) satisfying reflexivity, antisymmetry, and transitivity. In the multiparameter setting, we use the product partial order on \(\mathbb{R}^2\), given by
\[
(a,b)\leq(c,d) \iff a\leq c \text{ and } b\leq d.
\]
Then, an \textit{interval} \(A\) in \(P\) satisfies the following: \newline
\hspace*{1cm} 1) \(a,c\in A\) \(\&\) \(a\leq b\leq c \implies b\in A.\) \newline
\hspace*{1cm} 2) For any \(a,c\in A\), there exists a finite sequence \(a=b_0,b_1,\ldots,b_n=c\) with \newline
\hspace*{1.5cm} each consecutive \(b_i,b_{i+1}\), \(0\leq i<n\), being comparable under the partial \newline
\hspace*{1.5cm} ordering.
\end{defn}

Multiparameter persistence modules are indexed by partially ordered parameter spaces, where increasing one or both filtration parameters produces inclusion maps between simplicial complexes.
\begin{defn} If $P = S_1 \times S_2 \times ... \times S_n $ with $S_i$ ordered for all $i$ and $F$ a filtration on $P$, then $F$ is a \textit{multiparameter filtration}. \end{defn}The following commutative diagram showcases the extension of the notion of a filtration into two dimensions, where $F_{i,j}$ denotes the simplicial complex obtained after applying filtration thresholds $i,j$ along two distinct parameters. 
\[
\xymatrix{
\vdots & \vdots & \vdots \\
F_{0,2} \ar@{^{(}->}[r] \ar@{^{(}->}[u] 
& F_{1,2} \ar@{^{(}->}[r] \ar@{^{(}->}[u] 
& F_{2,2} \ar@{^{(}->}[r] \ar@{^{(}->}[u] 
& \cdots \\
F_{0,1} \ar@{^{(}->}[r] \ar@{^{(}->}[u] 
& F_{1,1} \ar@{^{(}->}[r] \ar@{^{(}->}[u] 
& F_{2,1} \ar@{^{(}->}[r] \ar@{^{(}->}[u] 
& \cdots \\
F_{0,0} \ar@{^{(}->}[r] \ar@{^{(}->}[u] 
& F_{1,0} \ar@{^{(}->}[r] \ar@{^{(}->}[u] 
& F_{2,0} \ar@{^{(}->}[r] \ar@{^{(}->}[u] 
& \cdots
}
\]

In our case, the most practical of such filtrations is the \textit{bifiltration}, utilizing two increasing parameters instead of one. Single-parameter filtrations such as $VR_\varepsilon$ rely on the expansion of a single $\varepsilon$, while bifiltrations consider $\varepsilon_1, \varepsilon_2$. Thus, simplices are only included in states of the bifiltration if two conditions are met. Multiparameter persistence modules follow the same natural extension. 

Unfortunately, however, complications arise with attempts to extend barcodes to multiparameter persistence. In the multiparameter case, there is generally no comparably simple barcode representation, which makes summarization and visualization substantially more difficult.

One alternative is the study of multiparameter persistence through one-dimensional affine slices to produce \textit{fibered barcodes} \cite{BotnanLesnick2022,Vipond2018}. More generally, multipersistence often requires alternative summaries, one-dimensional restrictions, or vectorized invariants rather than a direct analogue of the ordinary persistence barcode. Common invariants examined include multi-graded Betti numbers \cite{Knudson2008}, rank invariants \cite{CarlssonZomorodian2009}, and the Hilbert function \cite{Hilbert1890,CarlssonZomorodian2009,BotnanLesnick2022}. We use the same intuition of reducing multiparameter structure to one-dimensional filtrations but do not compute full fibered barcode invariants due to added computational cost and limited relevance to the main forecasting comparison. 
\subsubsection{Interleaving Distance}
Interleaving distance compares persistence modules directly: it measures the smallest filtration-scale shift needed for two modules to map into one another in a way that is compatible with their internal structure maps.

\begin{defn} Let $A,B$ be persistence modules over $\mathbb{R}^n$, and let $\varepsilon \geq 0$. Write $B(\varepsilon)$ for the shifted module defined by
\[
B(\varepsilon)_x = B_{x+\varepsilon\mathbf{1}},
\]
where $\mathbf{1}=(1,\ldots,1)$, so that $x+\varepsilon\mathbf{1}$ shifts each filtration coordinate forward by $\varepsilon$. The persistence modules \(A\) and \(B\) are said to be \(\varepsilon\)-interleaved if there exist morphisms
\[
f:A\to B(\varepsilon),
\qquad
g:B\to A(\varepsilon),
\]
such that the following diagrams commute:
\[
\begin{tikzcd}
A \arrow[r,"f"] \arrow[rr,bend right=35,"\varphi_A^{2\varepsilon}"'] 
& B(\varepsilon) \arrow[r,"g(\varepsilon)"] 
& A(2\varepsilon)
\end{tikzcd}
\]
and
\[
\begin{tikzcd}
B \arrow[r,"g"] \arrow[rr,bend right=35,"\varphi_B^{2\varepsilon}"'] 
& A(\varepsilon) \arrow[r,"f(\varepsilon)"] 
& B(2\varepsilon).
\end{tikzcd}
\]
We then define the \textit{interleaving distance} to be
\[
d_{\mathrm{Int}}(A,B)
=
\inf\{\varepsilon \geq 0 : A \text{ and } B \text{ are } \varepsilon\text{-interleaved}\}.
\] \end{defn}

Thus, small \(d_{\mathrm{Int}}(A,B)\) indicates that \(A\) and \(B\) contain topological features that appear and persist at similar filtration scales, while large \(d_{\mathrm{Int}}(A,B)\) indicates that the modules only become compatible after a substantial shift in scale.

\section{Methodology}
\subsection{Data}
A key question examined in this paper is whether compression of global prediction market belief states through PCA and persistent homology improves forecasting of binary market outcomes beyond market-implied probabilities, as measured by YES-contract prices. We compare these compressed representations against the raw market-implied baseline and standard statistical and machine-learning baselines. The empirical universe consists of macro, crypto, and crypto-adjacent Polymarket contracts, selected because they contain many related markets exposed to common latent drivers. Macro markets are included to introduce richer outcome diversity, while crypto and crypto-adjacent markets provide related families of contracts whose probabilities plausibly share underlying structure. Table \ref{tab:data_universe_summary} summarizes the resulting market universe and data sources.

\begin{table}[htbp]
\centering
\small
\renewcommand{\arraystretch}{1.18}
\begin{tabularx}{0.95\textwidth}{lX}
\hline
\textbf{Characteristic} & \textbf{Value} \\
\hline
Platform & Polymarket \\
Market universe & Macro, crypto, and crypto-adjacent resolved binary markets \\
Selection criteria & Binary resolution; valid YES token; resolved outcome; sufficient price observations; quality checks; sufficient overlap with analysis period \\
Total markets & 171 \\
Outcome split & 152 NO, 19 YES \\
Primary forecast variable & Market-implied YES probability from YES-contract price \\
Primary data representation & Hourly active-set panel \\
Data sources & Gamma API; Data API; CLOB API \\
Pipeline outputs & Metadata; selected universe; YES prices; CLOB-sampled prices; long-form price data; hourly active-set panels \\
\hline
\end{tabularx}
\caption{Summary of the prediction market universe and data pipeline.}
\label{tab:data_universe_summary}
\end{table}

\begin{table}[H]
\centering
\small
\renewcommand{\arraystretch}{1.15}
\begin{tabular}{lc}
\hline
\textbf{Market family} & \textbf{Number of markets} \\
\hline
Macro Fed / rates & 96 \\
BTC price & 56 \\
ETH price & 14 \\
Crypto policy & 2 \\
MicroStrategy & 2 \\
Crypto regulation & 1 \\
\hline
\textbf{Total} & \textbf{171} \\
\hline
\end{tabular}
\caption{Exclusive market-family counts in the final universe.}
\label{tab:market_family_counts}
\end{table}
Evaluation is conducted fold-by-fold over chronological splits. A \textit{chronological fold} is a train-test time-ordered split.

\subsection{Active-Set Belief States}
An important fact to note when dealing with prediction markets is that markets do not necessarily exist at the same time. Hence, a standard panel which pairs every timestamp with every market implicitly treats inactive markets as missing observations. If a market has not yet opened or has already resolved, then its absence from the panel is entirely structural, and treating such entries as ordinary missing values risks forcing PCA or PH to train on market availability instead of the evolution of market belief.

To avoid this lapse, we construct an \textit{active-set} belief representation. Let $M$ denote the final universe of markets. For timestamp $t$, define
\[
\mathcal A_t=\{i\in M: \text{market }i\text{ has opened and has not yet resolved at time }t\}.
\]

Because the active markets at a timestamp come from several related but distinct groups of contracts, we also organize markets by family.

\begin{defn}
A \textit{market family} \(f\in\mathcal F\) is a category of markets grouped by the primary economic or event driver of the contract, such as BTC price, ETH price, Macro Fed/rates, crypto policy, MicroStrategy, or crypto regulation.
\end{defn}
A supervised row is included if and only if $i\in\mathcal A_t$, so each observation can be represented by a pair $(i,t)$ with $i\in\mathcal A_t$. For each such pair, we record market probability $p_{i,t}$, the binary outcome $Y_i$, the timestamp $t$, the market family, and the active market count
\[
A_t=|\mathcal A_t|\in\mathbb{Z}_{\geq 0}.
\]
Our final active-set supervised dataset contains 532,979 market-time observations across 171 resolved binary markets and 14,464 hourly timestamps. The active market count has median 34, ranging from 17 to 68. The target variable is the final market resolution:
\begin{equation}
  Y_i=\begin{cases}
    1, & \text{if market $i$ resolves YES}.\\
    0, & \text{otherwise}.
  \end{cases}
\end{equation}

For each market family $f$, define the active family set
\[
\mathcal A_{f,t}=\{i\in\mathcal A_t: i\text{ belongs to family }f\}.
\]
We then compute
\[
c_{f,t}=|\mathcal A_{f,t}|,
\]
\[
\mu_{f,t}=\frac{1}{c_{f,t}}\sum_{i\in\mathcal A_{f,t}}p_{i,t},
\]
\[
m_{f,t}=\operatorname{median}\{p_{i,t}:i\in\mathcal A_{f,t}\},
\]
and
\[
s_{f,t}=\operatorname{std}\{p_{i,t}:i\in\mathcal A_{f,t}\}.
\]
These are the active count, mean active probability, median active probability, and probability dispersion for family $f$ at time $t$. The combination of these features leaves us with the desired family-level belief-state vector
\[
F_t=(c_{f,t},\mu_{f,t},m_{f,t},s_{f,t})_{f\in\mathcal F}.
\]
In our final dataset, this computation outputs 24 family-state features across six market families. The primary advantage of this construction is that $F_t$ summarizes the global market belief state without requiring the same markets to exist at every timestamp. It is less granular than the raw, market-level panel, but it is cleaner for compression because it avoids treating inactive markets as missing values. Hence, $F_t$ becomes the canonical belief-state object used for PCA and topology comparisons.

When a family-level feature is missing because no market from that family is active at a timestamp, imputation is learned only from the training timestamps of the current chronological fold. The fitted imputer is then applied to later test timestamps, so future observations do not enter the construction of past belief states. 
\subsection{PCA Benchmark}
The PCA benchmark is the first compressed-state forecasting model evaluated in this paper. It is intended to test whether the family-level belief-state vector $F_t$ contains relevant information beyond individual market-implied probabilities $p_{i,t}$. This benchmark also serves as the primary linear, non-topological comparison for later PH models. 
For each chronological fold we begin with the family-state matrix whose rows are timestamps $t$ and columns are the 24 active-set family features. Thus, for each test timestamp $t$, the corresponding PCA coordinates $z_t^{PCA}$ are produced using principal components of past data. 
The PCA features are then merged onto the supervised market-time dataset. Each row $(i,t)$ receives the same timestamp-level PCA state $z_t^{PCA}$ but keeps its own market probability $p_{i,t}$, active market count $A_t$, and final outcome $Y_i$. This creates the PCA forecasting dataset $$(i,t,p_{i,t}, A_t, z_t^{PCA}, Y_i).$$ 
The raw market baseline $\hat{p}_{i,t}$ is $p_{i,t}$, and the PCA-enhanced benchmark is a logistic regression model of the form $$\hat{p}_{i,t} = \Lambda(\beta_0 + \beta_1p_{i,t} + \beta_2A_t + \gamma^{\mathsf{T}}z_t^{PCA}),$$
where $\Lambda(x) = \frac{1}{1+e^{-x}}$. 
The individual market probability is included in every PCA model, meaning that PCA is being tested on whether the compressed global belief state improves the forecast after the inclusion of market price. Several PCA configurations were considered, including fixed component counts and variance-threshold choices. The family-level PCA representation was selected as it improved both Brier Score and log loss. Table \ref{tab:pca_benchmark} compares the market-implied probability baseline to the final two-component family-level PCA logistic benchmark. \newline 
\begin{table}[H]
\centering
\begin{tabular}{lcc}
\hline
\textbf{Model} & \textbf{Brier Score} & \textbf{Log Loss} \\
\hline
Market-implied probability & 0.0489 & 0.1673 \\
Family-level PCA logistic model & 0.0432 & 0.1565 \\
\hline
\end{tabular}
\caption{Forecasting performance of the market-implied probability baseline and the best family-level PCA benchmark.}
\label{tab:pca_benchmark}
\end{table}
This corresponds to an approximate $11.7\%$ Brier improvement and $6.5\%$ log loss improvement over the market-implied baseline (further statistical analysis is examined in Section 4). Therefore, the family-level PCA model becomes the linear benchmark that later topology-enhanced models seek to improve upon. 
It is important to note that this model was selected due to forecasting performance, not explained variance. Although higher-dimensional PCA specifications did capture more variation from the family-state matrix, they did not improve out-of-sample forecasting, with the two-component PCA model achieving the best performance among the family-level PCA variants. The 5-PC model performed slightly worse in both metrics, and all models with 10 PCs or more substantially worsened log loss. This suggests that the vast majority of forecasting-relevant linear structure is captured by the leading two PCs, while additional components largely introduce noise or less stable variation.  
\subsection{Local Topology}
The initial single-parameter topological exploration of the paper is designed as an enhancement to PCA, since persistent homology is most commonly used in supervised learning as a feature representation whose effectiveness depends on how persistence information is vectorized and integrated into a downstream model \cite{PunXiaLee2018}. We decide to focus on local topology construction in order to determine whether the geometry of a market's correlated neighborhood contains information about that market's eventual resolution beyond the implied probability and PCA-compressed global state. 
For each market $i$ we construct a local neighborhood of related markets using training period correlations. The neighborhood $N_k(i)$ is defined to be the set of $k$ markets most closely correlated with $i$, where the main specification uses $k = 20$. Each market $i$ is analyzed against $N_k(i)$, 
and topology is constructed from a correlation-distance Vietoris-Rips filtration. For two markets $i,j \in N_k(i)$, let $\rho_{ij}$ denote their training-period correlation. We define distance as \newline \centerline{$d_{ij} = 1 - |\rho_{ij}|.$} \newline
$d_{ij}$ is small when there is a strong correlation between $i,j$ in either direction, and large when the relation is weak. We take the absolute value of $\rho_{ij}$ since both positive and negative relationships are structurally significant. Either might provide forecasting insight since both co-movement and inverse movement can reveal shared exposure to an underlying factor.
For each $N_k(i)$ the distance matrix
\[
D_i=\{d_{ij}\}
\]
defines a finite dissimilarity space, a symmetric pairwise dissimilarity derived from correlations, on which we compute a Vietoris-Rips filtration. Markets that are strongly related enter the same connected structure early in the filtration, while weaker relationships appear at larger $\varepsilon$. The resulting filtration is denoted
\[
VR_\varepsilon(N_k(i),d),
\]
where $\varepsilon$ ranges over the correlation-distance scale. PH is then computed on this filtration in dimensions \(H_0\) and \(H_1\). The \(H_0\) persistence diagram records the birth and merging of connected components as the filtration scale increases, while the \(H_1\) persistence diagram records one-dimensional cycles that appear and later disappear. Each local persistence diagram is then converted into numerical features for use by the forecasting models. For both \(H_0\) and \(H_1\), we extract the number of finite persistence intervals, total persistence, maximum persistence, and persistence entropy. These quantities summarize how much connected-component or cycle structure appears in the local market neighborhood and how long that structure survives across filtration scales. A neighborhood whose diagram contains only many short-lived intervals is interpreted differently from one with longer-lived intervals, since longer persistence suggests structure that is less likely to arise from small perturbations in the filtration. 
Topological features are then attached to the same supervised market-time rows utilized in the PCA benchmark, so each observation has the form $$(i,t,p_{i,t}, A_t, z_t^{PCA}, T_i, Y_i),$$ 
for $T_i$ denoting the local topology feature vector for market $i$. When topology is allowed to vary with time, the notation becomes $T_{i,t}$ but the key object remains a topological summary of the market's correlated local neighborhood. The primary topology-enhanced logistic model is defined $$\hat{p}_{i,t}^{PCA + TDA} = \Lambda(\beta_0 + \beta_1p_{i,t} + \beta_2A_t + \gamma^{\mathsf{T}}z_t^{PCA} + \theta^{\mathsf{T}}T_i).$$
Here, \(p_{i,t}\) captures the individual market's current belief, \(A_t\) controls for variation in the size of the active market universe, \(z_t^{PCA}\) captures the global linear belief state, and \(T_i\) captures local nonlinear neighborhood geometry. The coefficient vector \(\theta\) measures whether the topological summaries add forecasting information beyond market price and PCA state after accounting for the other features in the model. For robustness, we compare topology-enhanced models against models containing simple neighborhood controls, including neighborhood mean probability, probability dispersion, average correlation strength, and correlation variability. These controls summarize first-order properties of the local neighborhood without using persistent homology. This tests whether PH contributes information about the shape of local market relationships or is simply repackaging non-topological facts about correlation strength, dispersion, or neighborhood density. 
\subsection{Multipersistence Extension}
The local $VR_\varepsilon$ construction above utilizes a single scale parameter $\varepsilon$ for filtration. Within the context of prediction markets, such filtration is limiting as two markets can be close in one sense (i.e. price movement), while differing substantially in behavior not explained by global PCA state. 
To address this, we construct a multipersistence-adjacent extension of the local topology model. We utilize the intuition of varying two filtration criteria, but do not compute entire multiparameter persistence modules or full barcode invariants due to unnecessary computational complexity. 
The first filtration coordinate $\alpha$ is the same correlation distance utilized on the local $VR_\varepsilon$ construction. For markets $i,j$, again let $\rho_{ij}$ denote their training period correlation, and define $$\alpha_{ij} = 1 - |\rho_{ij}|,$$ 
functioning identically to $\varepsilon$ in the section above. 
The second filtration parameter is designed to capture nonlinear or residual structure not explained by PCA. Within each chronological fold, PCA is fit on the training family-state matrix, and the PCA reconstruction is then subtracted from the original family-state representation to form residual behavior. Informally, if PCA captures dominant linear belief factors, the residual component measures what is left after the removal of those linear factors. Let $r_i$  denote this residual behavior for market $i$ after removal of the PCA-explained component. For markets $i,j$ define residual distance to be $$\delta_{ij} = ||r_i - r_j||.$$
Under these definitions, $\alpha_{ij}$ measures correlation-neighborhood proximity while $\delta_{ij}$ measures similarity in PCA-residual behavior. A bifiltration dependent upon these parameters includes a simplex $\sigma$ at parameter value $(a,b)$ if and only if all pairwise relationships in the simplex satisfy
\[
\alpha_{ij}\leq a \quad \text{and}\quad \delta_{ij}\leq b \qquad \forall i,j\in\sigma.
\]
Thus, the bifiltration is indexed by $(a,b)\in\mathbb{R}^2$, where $a$ controls correlation-distance scale and $b$ controls residual-distance scale. A group of markets enters the simplicial complex only if it is sufficiently connected in both raw correlation geometry and PCA-residual geometry. 
Due to computational complexity surrounding summarization of full multiparameter persistence in a forecasting model, we forgo computation of multiparameter Betti tables, rank invariants, and fibered barcode invariants. Instead, we utilize one-dimensional reductions of this two-parameter construction, with each reduction converting the bifiltration into an ordinary one-parameter filtration, after which standard \(H_0\) and \(H_1\) persistent homology features are computed. These reductions retain the influence of both filtration parameters by examining paths through the bifiltration, so the resulting topological summaries still reflect joint variation in correlation distance and PCA-residual distance while preserving the forecasting design of the local topology model. The resultant feature vector is denoted $T_i^{MP},$
or $T_{i,t}^{MP}$ in time-varying specifications. These features are similarly added to the PCA-enhanced forecasting model, which becomes $$\hat{p}_{i,t}^{PCA + MP} = \Lambda(\beta_0 + \beta_1p_{i,t} + \beta_2A_t + \gamma^{\mathsf{T}}z_t^{PCA} + \theta^{\mathsf{T}}T_i^{MP}).$$
This interpretation represents a stronger hypothesis than the original local VR-topology, asking whether topology becomes more informative when computed on information unexplained by PCA in addition to that explained by it. By utilizing multipersistence as a modeling principle, the empirical question becomes whether these complementary filtration criteria produce topological features that improve forecasting beyond market probability, PCA, and simple neighborhood controls. 
\subsection{Forecasting Models}
The forecasting task is defined at the market-time level, with each observation being a pair $(i,t)$ for market $i$ active at timestamp $t$. The target is the binary resolution $Y_i$, with each model outputting a probability forecast $\hat{p}_{i,t}.$ Informally, these models are evaluated based on whether additional compressed or topological features improve probabilistic forecasts beyond prior model iterations. 
\subsubsection{Market-Implied Probability}
The first model uses the market price directly as the forecast, taking the price to be the market-implied probability of event occurrence. $$\hat{p}_{i,t} = p_{i,t}.$$
This serves as the natural benchmark, since prediction market prices themselves are commonly taken as probability forecasts. Our hope is that models using PCA, topology, or nonlinear learning improve upon this baseline. 
\subsubsection{Logistic Regression + PCA}
The first fitted model is the PCA-enhanced logistic regression benchmark $$\hat{p}_{i,t}^{PCA} = \Lambda(\beta_0 + \beta_1p_{i,t} + \beta_2A_t + \gamma^{\mathsf{T}}z_t^{PCA}),$$
testing whether or not compressed global belief-state vectors contain predictive information beyond market price. This model serves as the primary linear compression benchmark. 
\subsubsection{Logistic Regression + PCA + Neighborhood Controls}
We then add simple local-neighborhood controls. $$\hat{p}_{i,t}^{PCA + C} = \Lambda(\beta_0 + \beta_1p_{i,t} + \beta_2A_t + \gamma^{\mathsf{T}}z_t^{PCA} + \eta^{\mathsf{T}}C_i),$$
where $C_i$ denotes non-topological neighborhood summaries including mean probability, probability dispersion, average correlation strength, and correlation variability. 
\subsubsection{Logistic Regression + PCA + Controls + Local VR Topology}
The next supplement to the above model is the local correlation-distance VR topology feature vector $T_i^{VR}.$ $$\hat{p}_{i,t}^{PCA + C + VR} = \Lambda(\beta_0 + \beta_1p_{i,t} + \beta_2A_t + \gamma^{\mathsf{T}}z_t^{PCA} + \eta^{\mathsf{T}}C_i + \theta^{\mathsf{T}}T_i^{VR}).$$
\subsubsection{Logistic Regression + PCA + Controls + Multipersistence-Inspired Topology}
We extend the local VR-topology features with multipersistence-inspired features $T_i^{MP}$ described in 3.5, which represents topology computed through one-dimensional reductions of a bifiltration. $$\hat{p}_{i,t}^{PCA + C + MP} = \Lambda(\beta_0 + \beta_1p_{i,t} + \beta_2A_t + \gamma^{\mathsf{T}}z_t^{PCA} + \eta^{\mathsf{T}}C_i + \theta^{\mathsf{T}}T_i^{MP}).$$
Again, this model looks to answer whether incorporation of PCA-residual structure in addition to correlation-neighborhood structure improves forecasting ability through additional information capture. 
\subsubsection{XGBoost + PCA + Controls}
To expand testing beyond linear logistic models, we include a nonlinear tree-based benchmark. $$\hat{p}_{i,t}^{XGB} = f_{XGB}(p_{i,t}, A_t, z_t^{PCA}, C_i).$$
This XGBoost model utilizes the same non-topological feature set as the \textit{PCA + Controls} model, but allows for nonlinear interactions among variables, thereby testing whether information that topology captures is nonlinear structure that XGBoost can already learn from PCA and controls. 
\subsubsection{XGBoost + PCA + Controls + Multipersistence-Inspired Topology} 
The next model incorporates our multipersistence-inspired topology to XGBoost through incorporation of the $T_i^{MP}$ vector:  
$$\hat{p}_{i,t}^{XGB} = f_{XGB}(p_{i,t}, A_t, z_t^{PCA}, C_i, T_i^{MP}),$$
asking whether multipersistence-inspired features add predictive information  after the downstream model is allowed to learn nonlinear interactions. 
\subsubsection{Random Forest + PCA + Controls}
The Random Forest benchmark uses the same non-topological feature set as the XGBoost controls model:
$$\hat{p}_{i,t}^{RF} = f_{RF}(p_{i,t}, A_t, z_t^{PCA}, C_i),$$
serving as a separate nonlinear comparison.

\subsubsection{Random Forest + PCA + Controls + Multipersistence-Inspired Topology}
The final model adds the multipersistence-inspired topology vector to the Random Forest feature set:
$$\hat{p}_{i,t}^{RF} = f_{RF}(p_{i,t}, A_t, z_t^{PCA}, C_i, T_i^{MP}).$$

This model tests whether the strongest topological feature construction continues to provide forecasting information under a nonlinear model that is not gradient-boosted. If the multipersistence-inspired features improve both XGBoost and Random Forest, then the result is less likely to be driven by a peculiarity of one downstream algorithm and more consistent with the interpretation that local topological structure contains independent predictive signal.

\begin{table}[htbp]
\centering
\footnotesize
\renewcommand{\arraystretch}{1.25}
\begin{tabular}{c p{0.23\textwidth} p{0.34\textwidth} p{0.30\textwidth}}
\hline
\textbf{\#} & \textbf{Model} & \textbf{Feature Set} & \textbf{Purpose} \\
\hline

1 & Market probability 
& $p_{i,t}$ 
& Raw prediction market baseline \\

2 & Logit + PCA 
& $p_{i,t}$, $A_t$, $z_t^{PCA}$ 
& Linear compressed belief-state benchmark \\

3 & Logit + PCA + controls 
& $p_{i,t}$, $A_t$, $z_t^{PCA}$, $C_i$ 
& Tests simple local-neighborhood information \\

4 & Logit + PCA + controls + VR topology 
& $p_{i,t}$, $A_t$, $z_t^{PCA}$, $C_i$, $T_i^{VR}$ 
& Tests local correlation-distance PH \\

5 & Logit + PCA + controls + multipersistence
& $p_{i,t}$, $A_t$, $z_t^{PCA}$, $C_i$, $T_i^{MP}$ 
& Tests two-criterion topology-inspired features \\

6 & XGBoost + PCA + controls 
& $p_{i,t}$, $A_t$, $z_t^{PCA}$, $C_i$ 
& Nonlinear non-topological benchmark \\

7 & XGBoost + PCA + controls + multipersistence 
& $p_{i,t}$, $A_t$, $z_t^{PCA}$, $C_i$, $T_i^{MP}$ 
& Nonlinear topology-enhanced benchmark \\

8 & Random Forest + PCA + controls
& $p_{i,t}$, $A_t$, $z_t^{PCA}$, $C_i$
& Bagged-tree nonlinear benchmark \\

9 & Random Forest + PCA + controls + multipersistence
& $p_{i,t}$, $A_t$, $z_t^{PCA}$, $C_i$, $T_i^{MP}$
& Non-boosted nonlinear topology benchmark \\

\hline
\end{tabular}
\caption{Forecasting models evaluated in the main comparison.}
\label{tab:forecasting_models}
\end{table}

\subsection{Evaluation}
Let \(\hat p_{i,t}\in[0,1]\) denote a model's predicted probability that market \(i\) resolves YES, and let \(Y_i\in\{0,1\}\) denote the realized outcome. The primary evaluation metrics are Brier Score,
\[
B=\frac{1}{N}\sum_{(i,t)}(\hat p_{i,t}-Y_i)^2,
\]
and log loss,
\[
L=-\frac{1}{N}\sum_{(i,t)}
\left[
Y_i\log(\hat p_{i,t})+(1-Y_i)\log(1-\hat p_{i,t})
\right].
\]
Brier Score measures squared probability error, while log loss penalizes confident wrong predictions more sharply. We also report AUC to measure ranking performance, though the main comparisons are based on Brier Score and log loss because the models are intended to produce calibrated probability forecasts.

Model comparisons are reported against the market-implied probability baseline and the PCA + controls benchmark. The main model-comparison tables report unweighted means of the fold-level scores. Calibration is examined separately in Section \ref{sec:calibration}, where predictions are grouped by probability bucket and compared to realized YES rates.
\section{Results}
\subsection{Dataset Summary}
The construction of our dataset object is outlined in 3.1. In summary, it consists of \begin{itemize}[noitemsep, topsep=0pt]
    \item 171 markets.
    \item 532,979 market-time observations.
    \item Date range: 08/06/2024--04/01/2026.
    \item YES rates of $10.1\%$ by market-time rows and $11.1\%$ by unique markets.
    \item Active market count ranging from 17 to 68, with median 34.
    \item Family distribution: 96 macro Fed/rates, 56 BTC price, 14 ETH price, 2 crypto policy, 2 MicroStrategy, and 1 crypto regulation market.
\end{itemize}
\subsection{Model Comparison}
The raw YES market price performs well as a baseline forecasting instrument, since it represents the aggregate belief of traders at time $t$, with additional volume compounding accuracy. However, it is not fully efficient in this dataset. The PCA model improves performance, indicating that family-level belief states contain useful information beyond that captured in individual market prices. Adding local neighborhood controls improves performance further, suggesting that the immediate surrounding environment of a market also contains relevant information. \newline
\begin{table}[H]
\centering
\small
\begin{tabular}{lccc}
\hline
\textbf{Model} & \textbf{Brier Score} & \textbf{Log Loss} & \textbf{AUC} \\
\hline
Market-implied probability & 0.0489 & 0.1673 & 0.9419 \\
Logit + PCA & 0.0432 & 0.1565 & 0.9436 \\
Logit + PCA + controls & 0.0362 & 0.1355 & 0.9492 \\
\hline
\end{tabular}
\caption{Baseline, PCA, and local-control forecasting performance. Higher Area Under Curve (AUC) indicates stronger ranking performance.}
\label{tab:pca_controls_comparison}
\end{table}
The improvement via the addition of controls is particularly important for motivating topology. If local controls were useless then it would be unlikely that local topological summaries contained helpful information, but the control model was able to show that local market structure carries signal. Hence, we continue to ask whether this local structure has a shape that summary statistics do not capture. 
\begin{figure}[htbp]
    \centering
    \includegraphics[width=0.85\linewidth]{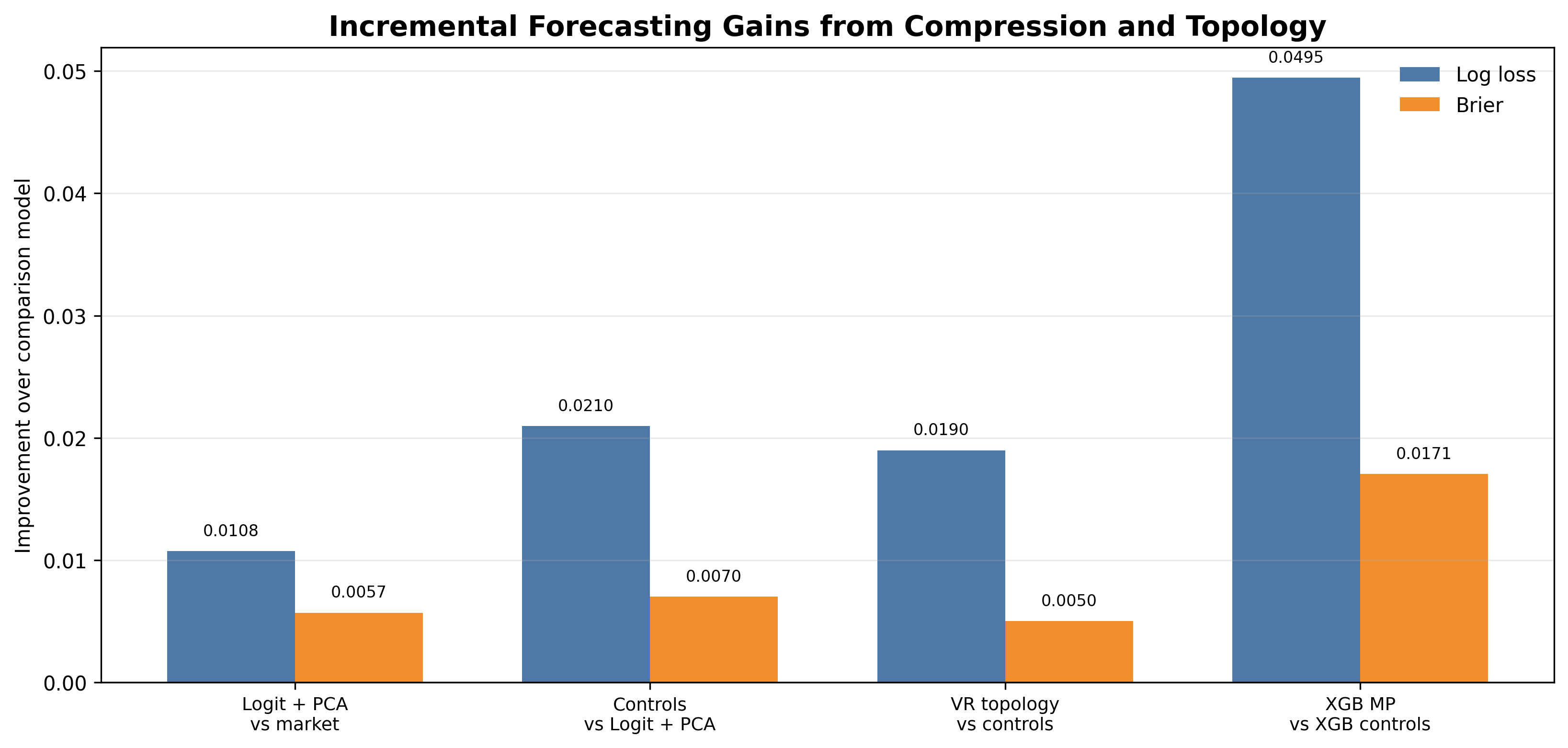}
    \caption{Incremental forecasting gains from adding PCA, local controls, and topology. Positive values indicate improvement over the comparison model.}
    \label{fig:incremental_model_gains}
\end{figure}
 Figure \ref{fig:incremental_model_gains} makes the staged improvement evident. PCA improves over market-implied baseline, local controls improve over PCA, and topology improves over the non-topological control model. 

 Improvement of the Vietoris-Rips model over the logistic PCA + Controls benchmark gives us our cleanest topological result, since the comparison is conservative. The $VR_\varepsilon$ model improves over a model that contains market probability, active market count, PCA state, and local controls, showing that topology is capturing structure of market relations across filtration scales. 
 \begin{table}[htbp]
\centering
\small
\begin{tabular}{lccc}
\hline
\textbf{Model} & \textbf{Brier Score} & \textbf{Log Loss} & \textbf{AUC} \\
\hline
Logit + PCA + controls & 0.0362 & 0.1355 & 0.9492 \\
Logit + PCA + controls + VR topology & 0.0311 & 0.1166 & 0.9629 \\
XGBoost + PCA + controls & 0.0402 & 0.1426 & 0.9712 \\
\hline
\end{tabular}
\caption{Vietoris-Rips topology comparison.}
\label{tab:vr_topology_comparison}
\end{table}
 \begin{figure}[htbp]
    \centering
    \includegraphics[width=0.85\linewidth]{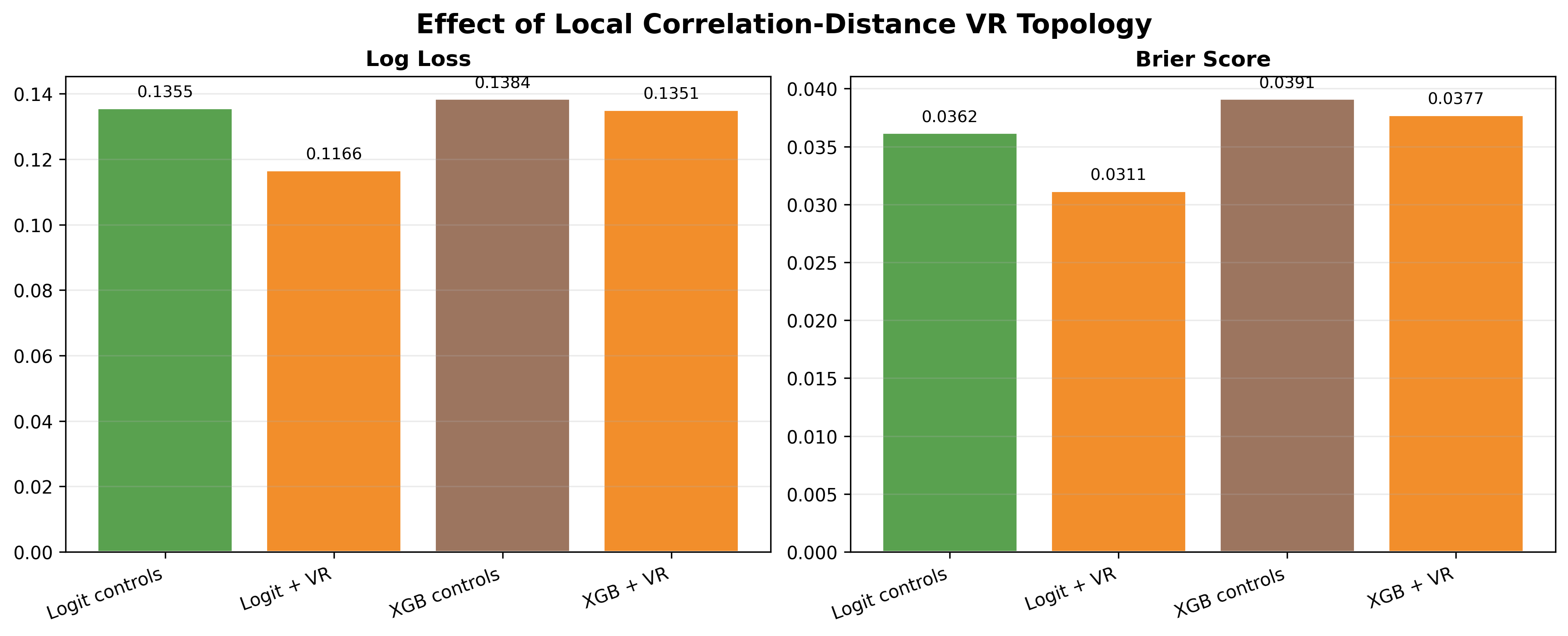}
    \caption{Effect of adding local correlation-distance Vietoris-Rips topology to logistic regression and XGBoost models.}
    \label{fig:vr_topology_model_comparison}
\end{figure}

The single-parameter VR-topology gain is larger for logistic regression than XGBoost. Logistic regression can only combine the input features linearly, so single-parameter persistent homology may help by giving the model nonlinear local geometry summaries. Conversely, XGBoost is already capable of learning nonlinear interactions between market probability, PCA features, active market count, and neighborhood controls, so the smaller gain suggests that some of the information contained in the VR summaries can be approximated through nonlinear interactions among non-topological variables. 
The multipersistence-inspired models extend this idea, allowing us to ask whether the local residual geometry following PCA compression contains useful forecasting structure. 
\begin{table}[H]
\centering
\small
\begin{tabular}{lccc}
\hline
\textbf{Model} & \textbf{Brier Score} & \textbf{Log Loss} & \textbf{AUC} \\
\hline
XGBoost + PCA + controls & 0.0402 & 0.1426 & 0.9712 \\
XGBoost + PCA + controls + MP & 0.0231 & 0.0932 & 0.9932 \\
Random Forest + PCA + controls & 0.0359 & 0.1544 & 0.9862 \\
Random Forest + PCA + controls + MP & 0.0126 & 0.0764 & 0.9944 \\
\hline
\end{tabular}
\caption{Multipersistence-inspired topology comparison.}
\label{tab:multipersistence_comparison}
\end{table}
\begin{figure}[htbp]
    \centering
    \includegraphics[width=0.85\linewidth]{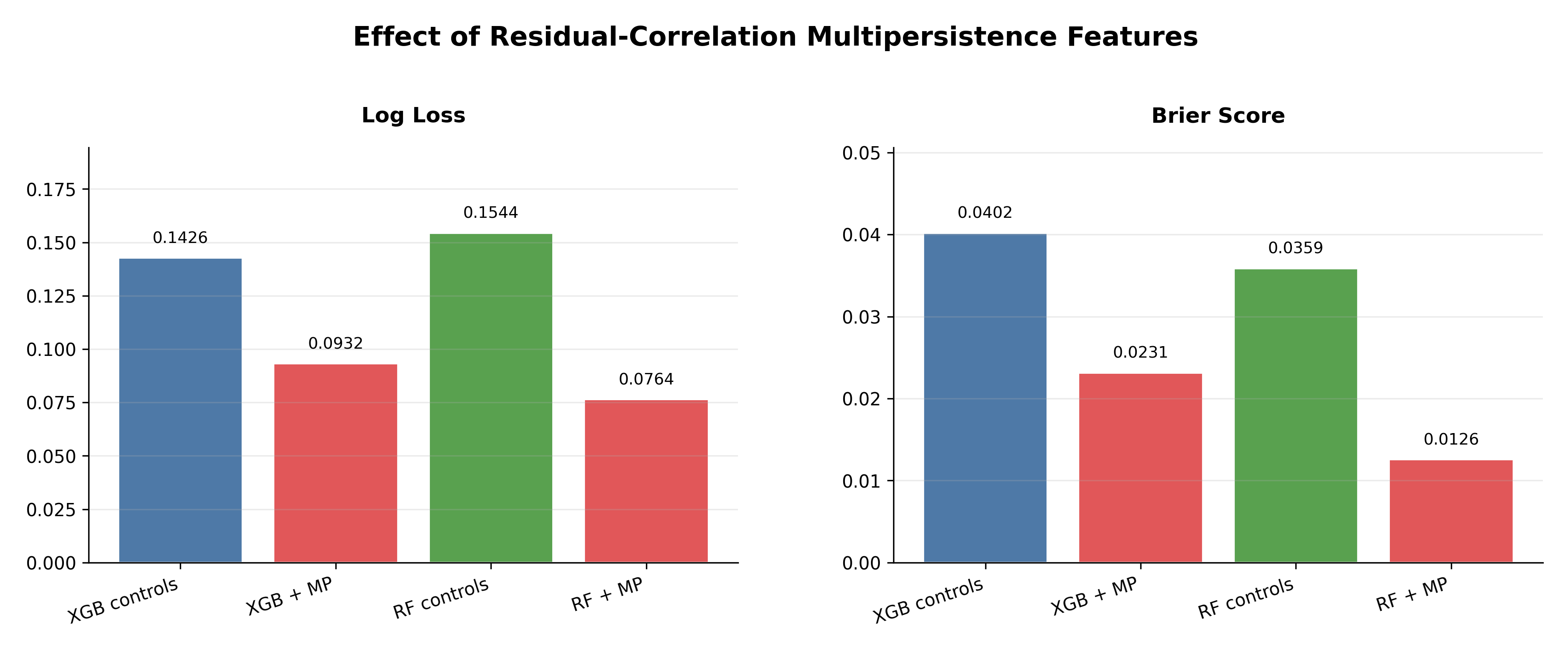}
    \caption{Effect of adding residual-correlation multipersistence features to XGBoost and Random Forest models.}
    \label{fig:multipersistence_model_comparison}
\end{figure}
Incorporation of multipersistence-inspired features into nonlinear models produces the strongest overall predictive results, with both XGBoost and Random Forest improving substantially with inclusion of PCA-residual and correlation distance scale factors. This suggests that the most useful topological signal is contained in the local residual structure left following PCA compression. Therefore, PCA and topology appear most useful when their roles are separated. 
We observe that a pairing of nonlinear learning with residual local-topology features yields the strongest forecasting performance, with the Random Forest multipersistence model having both the best Brier score and log loss.

However, the logistic VR result remains important since it more directly supports the claim that topology adds incremental information beyond PCA and controls. 
\begin{figure}[H]
    \centering
    \includegraphics[width=0.95\linewidth]{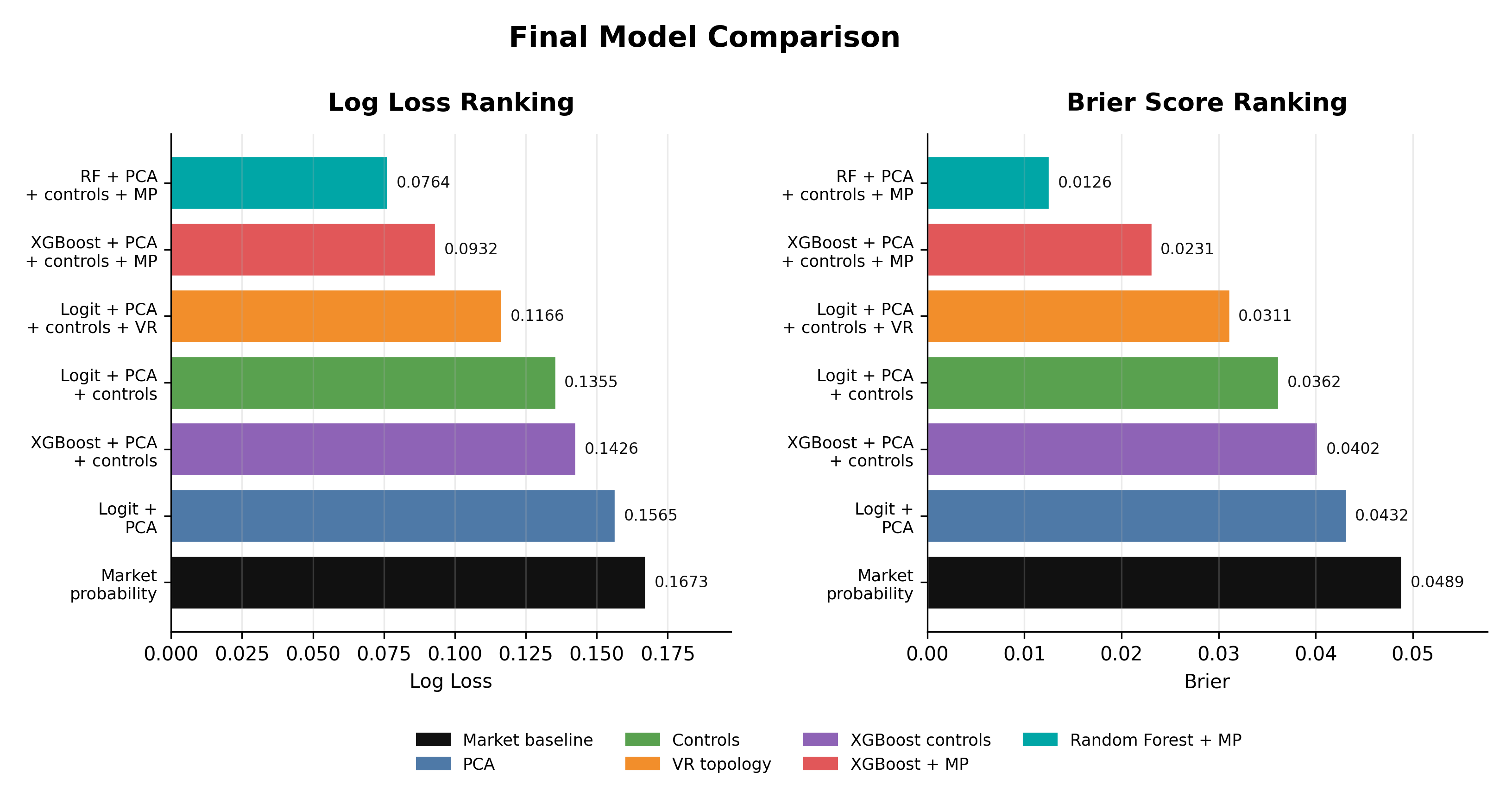}
    \caption{Final model ranking by log loss and Brier score. }
    \label{fig:final_model_ranking}
\end{figure}
We observe improvements in both log loss and Brier score when topology is added to the corresponding non-topological benchmarks. 
\subsection{Calibration}
\label{sec:calibration}
For a more comprehensive analysis of model behavior across the entire prediction range, we examine calibration as a central diagnostic. A well-calibrated model should assign probability $p$ to realized YES rates around $p.$ Mathematically, calibration requires
\[
\mathbb{P}(Y_i=1 \mid \hat p_{i,t}=p)=p,
\]
meaning that among observations assigned predicted probability \(p\), the realized YES frequency should equal \(p\). 
For evaluation, we group predictions into probability deciles, computing the average predicted probability and average realized YES rate for each model and decile. We also compute expected calibration error, given by the following formula:
\[
\operatorname{ECE}
=
\sum_{b=1}^{B}
\frac{n_b}{N}
\left|
\overline{p}_b-\overline{Y}_b
\right|,
\]
where $B$ denotes the number of probability bins, and within each $b \in \{1,\ldots,B\},$ $n_b$ gives number of observations, $\overline{p}_b$ is average predicted probability, and $\overline{Y}_b$ is the observed YES rate. Lower ECE is indicative of better calibration.

\begin{figure}[H]
    \centering
    \includegraphics[width=0.72\linewidth]{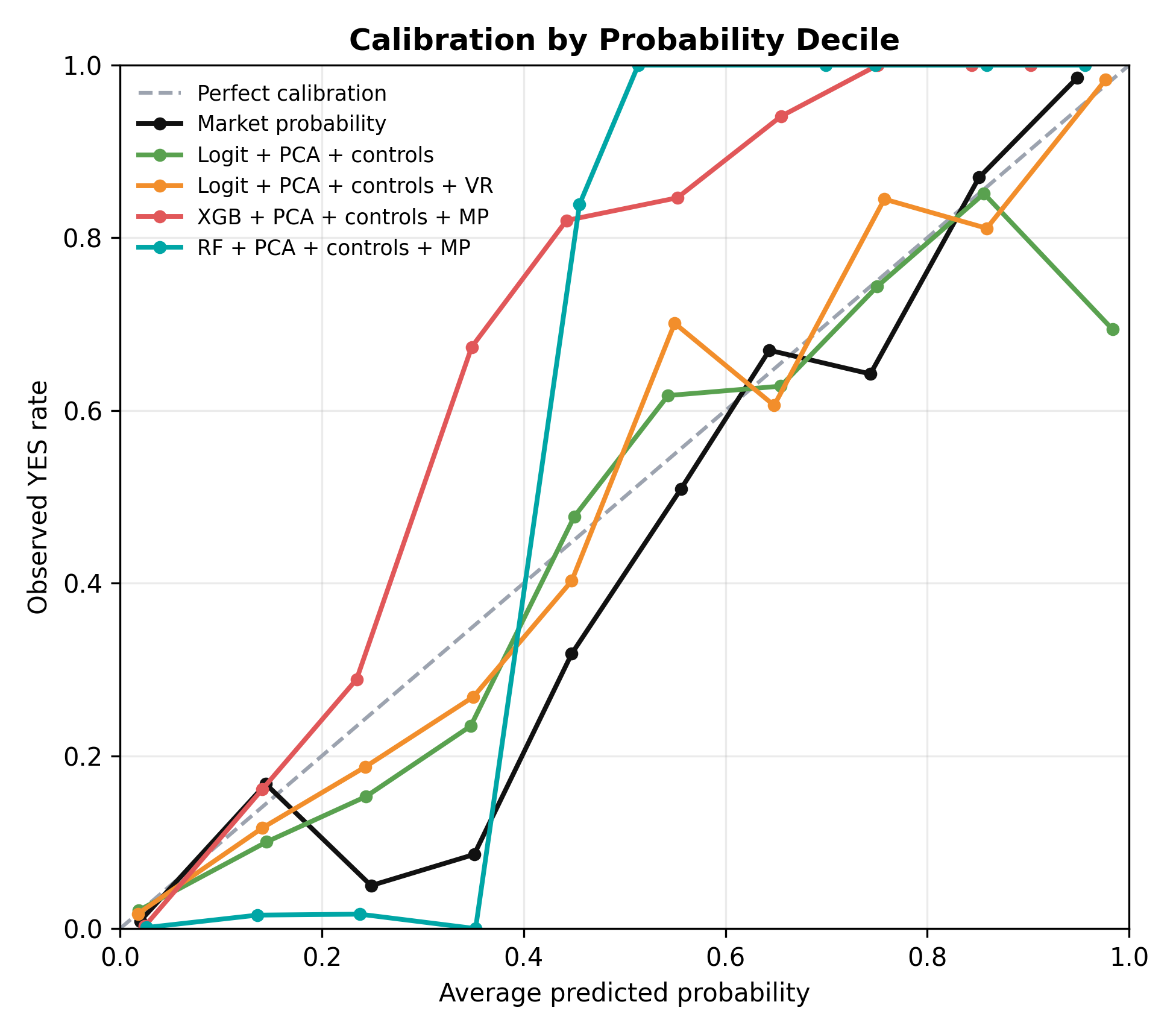}
    \caption{Calibration by probability decile.}
    \label{fig:final_calibration_by_decile}
\end{figure}

In Figure \ref{fig:final_calibration_by_decile} the dashed diagonal line represents perfect calibration, where average predicted probability equals observed YES rate. We see that logistic VR topology is the best calibrated model, which is consistent with its strong log loss performance. This figure also introduces a deeper intricacy to our initial results, as although Random Forest + Multipersistence and XGBoost + Multipersistence have strong Brier score and log loss performance, their calibration is less efficient than linear models. One possible explanation is that the nonlinear models improve discrimination by learning sharper interactions among PCA, neighborhood, and topological features, but these sharper probability adjustments lead to overconfident forecasts. The linear model is less flexible, but that restriction may also make its probability estimates smoother and better calibrated across deciles. 
\begin{table}[H]
\centering
\small
\begin{tabular}{lc}
\hline
\textbf{Model} & \textbf{Expected Calibration Error} \\
\hline
Logit + PCA + controls + VR topology & 0.0055 \\
Logit + PCA + controls & 0.0161 \\
Market-implied probability & 0.0374 \\
XGBoost + PCA + controls + multipersistence & 0.0403 \\
Random Forest + PCA + controls + multipersistence & 0.0562 \\
\hline
\end{tabular}
\caption{Expected calibration error for the main forecasting models.}
\label{tab:e10_calibration}
\end{table}
\begin{figure}[htbp]
    \centering
    \includegraphics[width=0.82\linewidth]{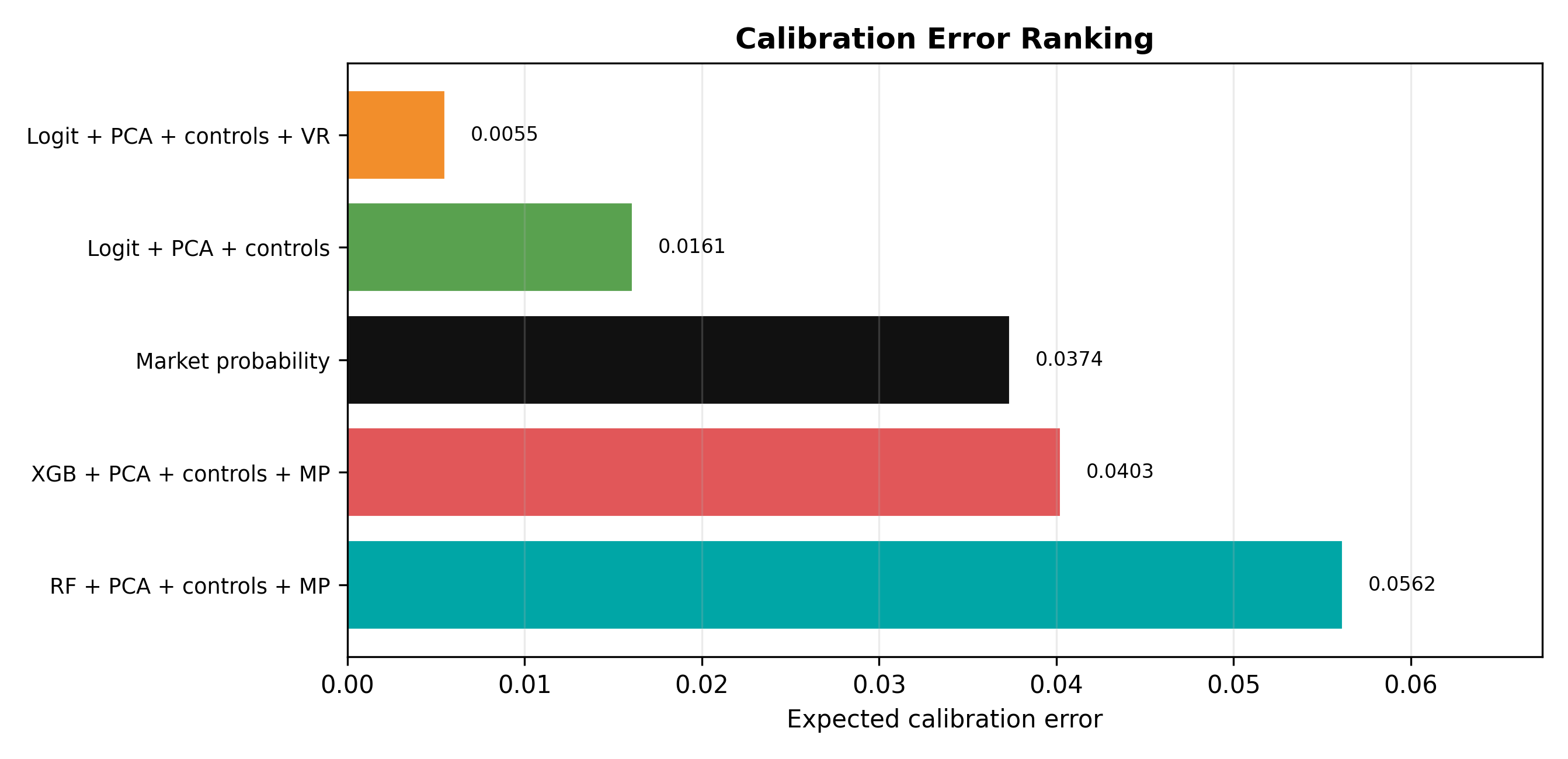}
    \caption{Expected calibration error ranking for the main forecasting models.}
    \label{fig:final_ece_ranking}
\end{figure}
Table \ref{tab:e10_calibration} and Figure \ref{fig:final_ece_ranking} summarize and rank the models based on ECE. Again, logistic PCA + Controls + VR topology had the lowest ECE, meaning it produces the most reliable probability estimates. Hence, our interpretation becomes layered. For pure predictive performance, nonlinear multipersistence models are strongest, while logistic VR topology has provided the most interpretable and calibrated probability forecasting. 
Since prediction markets are widely considered to be probability forecasts where positions are taken when the user believes such a forecast to be mispriced, and because the user can sell positions back to the market before resolution, calibration accuracy is an incredibly valuable result of the logit + PCA + controls + VR model. 
\subsection{Regime Analysis}
Through the aggregate model rankings, we have thus far observed that topology improves forecasting performance, but we have not yet evaluated where this improvement comes from. To do so, we inspect model gains across various market regimes including general domain, market family, and probability bucket. A positive gain indicates that the topological model has lower log loss than its corresponding non-topological benchmark. 
First we look at broad domain, separating crypto-related markets from macro / Fed rates markets. The logistic VR topology model improves in both domains, with a slightly larger gain in crypto markets. The nonlinear multipersistence models improve substantially in both domains, with Random Forest exhibiting especially large gains in macro markets. 
\begin{table}[H]
\centering
\small
\begin{tabular}{llccc}
\hline
\textbf{Model} & \textbf{Domain} & \textbf{Base Log Loss} & \textbf{Topology Log Loss} & \textbf{Gain} \\
\hline
Logit VR & Crypto & 0.1421 & 0.1198 & 0.0223 \\
Logit VR & Macro & 0.1227 & 0.1071 & 0.0156 \\
XGB MP & Crypto & 0.1234 & 0.0838 & 0.0396 \\
XGB MP & Macro & 0.1427 & 0.0964 & 0.0463 \\
RF MP & Crypto & 0.0962 & 0.0579 & 0.0383 \\
RF MP & Macro & 0.1822 & 0.0809 & 0.1013 \\
\hline
\end{tabular}
\caption{Topology log-loss gains by broad domain. }
\label{tab:regime_domain_gains}
\end{table}

The family-level split in Table \ref{tab:regime_family_gains} gives a more nuanced view, though the greatest observed gains in the MicroStrategy-related markets comprise a small minority of the overall dataset. 
\begin{table}[H]
\centering
\small
\begin{tabular}{llccc}
\hline
\textbf{Model} & \textbf{Family} & \textbf{Base Log Loss} & \textbf{Topology Log Loss} & \textbf{Gain} \\
\hline
Logit VR & MicroStrategy & 1.4458 & 1.2449 & 0.2009 \\
Logit VR & Macro Fed/rates & 0.1227 & 0.1071 & 0.0156 \\
Logit VR & BTC price & 0.0510 & 0.0406 & 0.0104 \\
Logit VR & ETH price & 0.0328 & 0.0267 & 0.0061 \\
\hline
XGB MP & MicroStrategy & 1.1334 & 0.6549 & 0.4785 \\
XGB MP & Macro Fed/rates & 0.1427 & 0.0964 & 0.0463 \\
XGB MP & ETH price & 0.0420 & 0.0342 & 0.0078 \\
XGB MP & Crypto policy & 0.0323 & 0.0246 & 0.0077 \\
\hline
RF MP & MicroStrategy & 0.2290 & 0.1074 & 0.1216 \\
RF MP & Macro Fed/rates & 0.1822 & 0.0809 & 0.1013 \\
RF MP & ETH price & 0.0729 & 0.0308 & 0.0421 \\
RF MP & Crypto policy & 0.0417 & 0.0083 & 0.0334 \\
\hline
\end{tabular}
\caption{Largest family-level topology gains by model comparison.}
\label{tab:regime_family_gains}
\end{table}

The probability bucket analysis in Table \ref{tab:regime_probability_bucket_gains} determines whether topological analysis is most informative in cases where markets are confident (corresponding to large or small YES prices), or when probabilities are closer to 50\%. 
\begin{table}[htbp]
\centering
\small
\begin{tabular}{llccc}
\hline
\textbf{Model} & \textbf{Probability Bucket} & \textbf{Base Log Loss} & \textbf{Topology Log Loss} & \textbf{Gain} \\
\hline
Logit VR & 50--75\% & 0.3219 & 0.2456 & 0.0763 \\
Logit VR & 25--50\% & 0.3228 & 0.2581 & 0.0647 \\
Logit VR & 75--90\% & 0.1217 & 0.0718 & 0.0499 \\
Logit VR & 10--25\% & 0.4572 & 0.4188 & 0.0384 \\
\hline
XGB MP & 50--75\% & 0.6104 & 0.4066 & 0.2038 \\
XGB MP & 90--95\% & 0.3616 & 0.1970 & 0.1646 \\
XGB MP & 75--90\% & 0.3806 & 0.2400 & 0.1405 \\
XGB MP & 10--25\% & 0.3444 & 0.2249 & 0.1195 \\
\hline
RF MP & 25--50\% & 0.3199 & 0.1450 & 0.1749 \\
RF MP & 10--25\% & 0.2572 & 0.0917 & 0.1655 \\
RF MP & 50--75\% & 0.2257 & 0.1263 & 0.0994 \\
RF MP & 5--10\% & 0.1485 & 0.0653 & 0.0832 \\
\hline
\end{tabular}
\caption{Largest topology gains by market probability bucket.}
\label{tab:regime_probability_bucket_gains}
\end{table}

The logistic VR model showed largest improvements in middle to high probability ranges, namely the 25-50\%, 50-75\%, and 75-90\% buckets, indicating that when market probabilities are close to 0 then there is little room for additional structure to improve the forecast. Topology appears most useful when market price contains enough uncertainty for the broader neighborhood structure to adjust the forecast. In these cases, correlated nearby markets can shift the model's estimate by indicating whether the individual contract is aligned with or deviating from the behavior of related contracts.

Altogether, this split regime analysis shows that topology gains are not concentrated within a single market domain or probability bucket. The strongest gains appear in market regimes where local structure plausibly matters: correlated BTC/ETH and Macro markets share broad latent drivers, MicroStrategy markets are tied to both firm-specific behavior and crypto exposure, and uncertain probability ranges leave more room for neighboring markets to adjust the forecast. Future work should test whether topology's predictive gains increase with shared market exposure and uncertainty.
\subsection{Significance Tests}
Because the dataset contains numerous market-time rows, observations are not ``independent'' in the ordinary sense. Rather, the same market appears repeatedly over time, with many markets clustered at similar timestamps. To account for this, we conducted clustered bootstrap tests and chronological fold-level improvement counts, looking to determine whether topological benefits persist when uncertainty is measured across markets, timestamps, and chronological folds. 

Gains are defined as $g = n - t,$
where $n$ is non-topological model loss and $t$ denotes topology-enhanced model loss. Hence, a positive log loss or Brier gain indicates an improvement upon the non-topological model by the topology-enhanced model. Here, improvement is a reduction in the relevant loss metric relative to the corresponding non-topological benchmark. Bootstrap gains are computed by pooling the paired market-time loss differences across test folds, with each market-time observation receiving equal weight. To assess uncertainty in these gains (sampling variation in the estimated improvement due to the particular markets and timestamps observed in the dataset) we compute clustered bootstrap confidence intervals (CIs) \cite{EfronTibshirani1993}. We provide market-clustered CIs to test whether improvement persists when resampling is performed across distinct markets. 
\begin{table}[H]
\centering
\footnotesize
\renewcommand{\arraystretch}{1.15}
\begin{tabular}{p{0.34\textwidth} c c c c}
\hline
\textbf{Comparison} & \textbf{LL Gain} & \textbf{Market CI} & \textbf{Brier Gain} & \textbf{Market CI} \\
\hline
Logit VR vs controls 
& 0.0179 & [0.0069, 0.0319] & 0.0052 & [0.0019, 0.0086] \\

XGB MP vs controls 
& 0.0440 & [0.0143, 0.0821] & 0.0152 & [0.0044, 0.0287] \\

RF MP vs controls 
& 0.0797 & [0.0603, 0.1030] & 0.0232 & [0.0143, 0.0318] \\
\hline
\end{tabular}
\caption{Market-clustered bootstrap confidence intervals for topology gains. }
\label{tab:market_clustered_significance}
\end{table}
The market-clustered intervals are positive for our three primary topology comparisons, the strongest result being for the Random Forest MP model, with slightly weaker gains for the XGBoost MP model, followed by the logistic VR model. This shows topological improvement is not simply a row-level artifact. 

Table \ref{tab:timestamp_clustered_significance} examines timestamp-clustered intervals, which are incredibly valuable since market shocks can impact many contracts simultaneously. A strong row-by-row model might be weaker than initially believed if its gains disappear when observations are clustered by timestamp. \begin{table}[H]
\centering
\footnotesize
\renewcommand{\arraystretch}{1.15}
\setlength{\tabcolsep}{4pt}
\begin{tabular}{p{0.25\textwidth} c c c c}
\hline
\textbf{Comparison} & \textbf{LL Gain} & \textbf{Timestamp CI} & \textbf{Brier Gain} & \textbf{Timestamp CI} \\
\hline
Logit VR vs controls 
& 0.0179 & [0.0175, 0.0183] & 0.0052 & [0.0051, 0.0053] \\

XGB MP vs controls 
& 0.0440 & [0.0434, 0.0446] & 0.0152 & [0.0150, 0.0154] \\

RF MP vs controls 
& 0.0797 & [0.0787, 0.0805] & 0.0232 & [0.0229, 0.0235] \\
\hline
\end{tabular}
\caption{Timestamp-clustered bootstrap confidence intervals for topology gains.}
\label{tab:timestamp_clustered_significance}
\end{table} 
Again, models retained positive gains for the same comparisons.

Consistency across chronological folds is yet another valuable metric, showcasing model robustness over time and lack of  performance isolation to particular folds. 
\begin{table}[H]
\centering
\small
\begin{tabular}{lcc}
\hline
\textbf{Comparison} & \textbf{Log-loss folds improved} & \textbf{Brier folds improved} \\
\hline
Logit VR vs controls & 16 / 17 & 15 / 17 \\
XGB MP vs controls & 17 / 17 & 17 / 17 \\
RF MP vs controls & 16 / 17 & 16 / 17 \\
\hline
\end{tabular}
\caption{Fold-level consistency of topology gains across chronological test folds.}
\label{tab:fold_level_significance}
\end{table}
As visible in Table \ref{tab:fold_level_significance}, each topological model improved both log loss and Brier score in most (if not all) folds, with the most consistent improvement coming from XGBoost MP. 
\section{Discussion}
\subsection{Interpretation}
The central result is that prediction market belief states contain compressible forecasting structure beyond individual market prices. The market-implied probability baseline is a strong predictive indicator as the literature suggests \cite{PolymarketAccuracy}, but the two-component family-level PCA model improves on the baseline's Brier score and log loss. Related markets do not move independently, sharing underlying factors which can be compressed into a low-dimensional belief state and used for outcome forecasting. A plausible mechanism for the PCA improvement is that individual market prices respond incompletely or unevenly to broad information shocks. For example, macro-rate expectations, crypto regulation news, or Bitcoin price movement may affect many contracts at once, but not every individual market necessarily incorporates this information at the same speed or magnitude. The PCA state therefore functions as a background regime variable, indicating whether the family-level market environment is moving in a direction consistent with or against a given contract's price. In this interpretation, PCA improves forecasting because it contextualizes that price within the contemporaneous movement of related markets.

We utilize persistent homology as a local-structure enhancement to PCA, with the neighborhood Vietoris-Rips model improving upon the logistic benchmark, and multipersistence-inspired utilization of two filtration scales (correlation distance and PCA residual distance) producing the strongest nonlinear results when coupled with XGBoost and Random Forest models. This multipersistence result suggests that PCA and topology capture inherently distinct market information, with PCA obtaining dominant global linear factors while topological features summarize nonlinear relationships between correlated markets. One conjecture about the cause of this is that the topological features are detecting market-specific disagreement patterns. When a local neighborhood contains markets that are mutually related but not moving in a purely uniform direction, persistent cycles can arise from inconsistent pairwise relationships among those markets. Such configurations may occur when traders have priced in a broad theme, such as crypto momentum or Fed policy, but disagree on how that theme maps onto specific contracts. The forecasting gain may therefore come from identifying neighborhoods where the individual market price is surrounded by structured tension. This would explain why topology is useful after PCA and controls, since PCA captures the dominant shared direction of movement, while topology detects residual relational structure within the local neighborhood.

The calibration results show that despite the strongest nonlinear models achieving the lowest Brier scores and log loss (and hence the best observed predictive performance by our chosen metrics), the logistic single-parameter VR topology model had the most accurate probability forecasting. Because prediction markets are probability-producing systems, calibration is equally important to forecasting accuracy. 
\subsection{Limitations}
There are three important limitations of this paper to consider. Firstly, within the market universe we analyzed there is a class imbalance, with substantially more NO resolutions than YES. However, this is a fact of prediction markets as a whole, with roughly 74\% of markets on Polymarket resolving to NO and 26\% resolving YES \cite{PolymarketAccuracy}, since there typically are more ways for events \textit{not} to happen than \textit{to} happen. Additionally, we account for this imbalance through incorporation of proper scoring rules, calibration checks, and chronological evaluation. 

For robustness, we run all models on a balanced market universe, with an even split of YES and NO-resolving markets. Table \ref{tab:balanced_universe_model_results} depicts these results.
\begin{table}[htbp]
\centering
\small
\begin{tabular}{lccc}
\hline
\textbf{Model} & \textbf{Brier Score} & \textbf{Log Loss} & \textbf{AUC} \\
\hline
Market probability & 0.1660 & 0.5204 & 0.8161 \\
Logit + PCA + controls & 0.1318 & 0.4008 & 0.8854 \\
Logit + PCA + controls + VR & 0.1185 & 0.3604 & 0.9050 \\
XGB + PCA + controls & 0.1189 & 0.3789 & 0.9125 \\
XGB + PCA + controls + MP & 0.0399 & 0.1727 & 0.9900 \\
RF + PCA + controls & 0.0802 & 0.2651 & 0.9524 \\
RF + PCA + controls + MP & 0.0143 & 0.0481 & 0.9916 \\
\hline
\end{tabular}
\caption{Balanced-universe forecasting performance.}
\label{tab:balanced_universe_model_results}
\end{table}

The balanced-universe results follow the same broad pattern as the natural-universe results, but with larger separations between the baseline, PCA/control models, and topology-enhanced nonlinear models. In particular, the strongest model remains RF + PCA + controls + MP, which reduces log loss from the market baseline of $0.5204$ to $0.0481$ and achieves the lowest Brier score of $0.0143$.

Figure \ref{fig:balanced_vs_natural_topology_gains} depicts the topology gains in the natural and balanced universes, with a substantial increase in topological benefit seen in the balanced universe. This indicates that topological analysis may provide even more drastic benefit in market universes with more even historical resolution splits. 
\begin{figure}[H]
    \centering
    \includegraphics[width=0.95\linewidth]{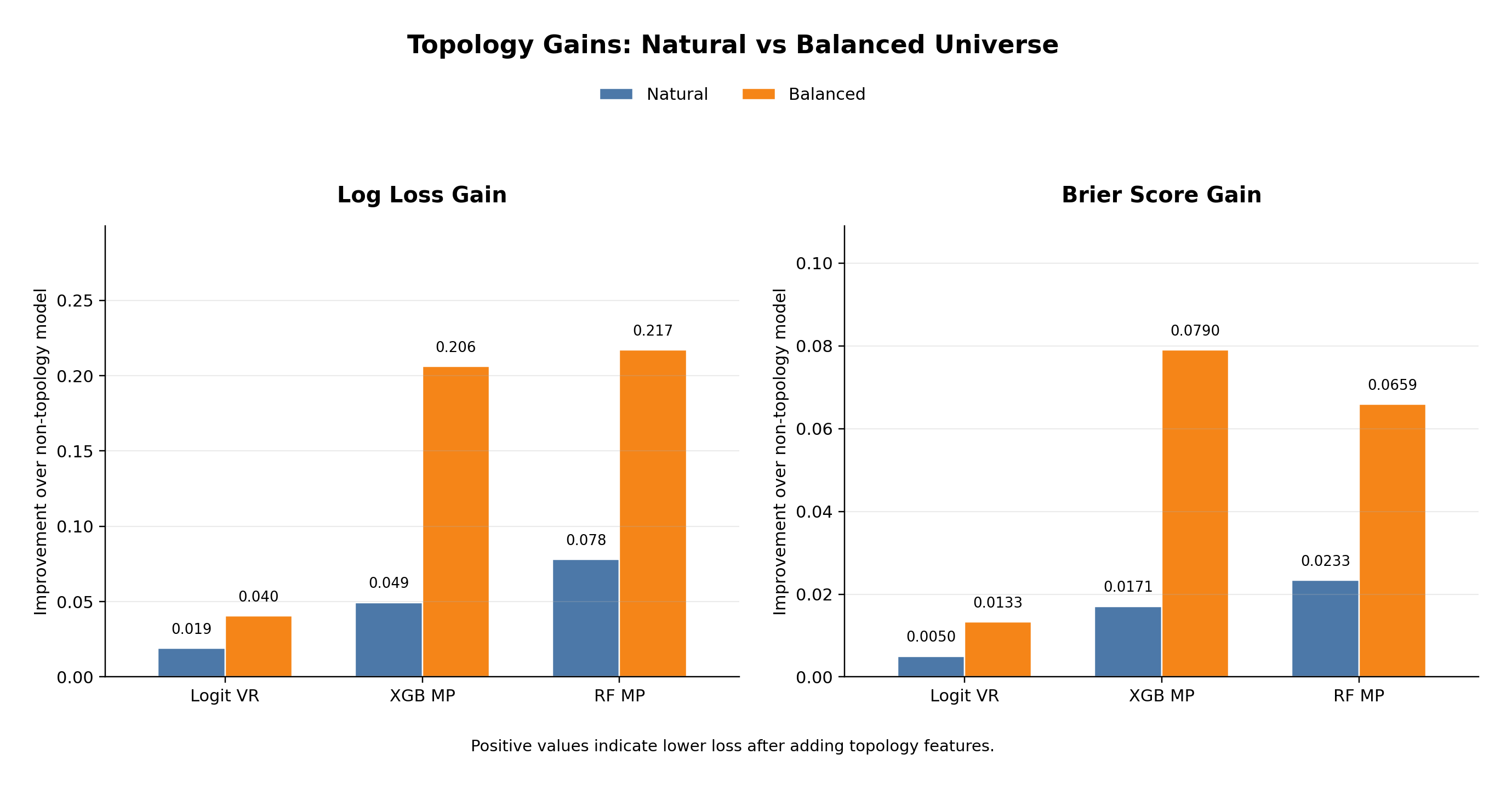}
    \caption{Topology gains on the natural and balanced market universes.}
    \label{fig:balanced_vs_natural_topology_gains}
\end{figure}

The second limitation is that our market universe is centered around macro, crypto, and crypto-adjacent markets. This is an intentional choice since these markets share latent drivers, but results should not be generalized automatically to all prediction markets. Markets with less underlying commonality may not benefit to the same degree from PCA or topological analysis. 

Finally, we utilize Vietoris-Rips topology and multipersistence-inspired reductions, but do not compute entire multiparameter persistence modules or rank invariants in order to keep features computationally manageable and usable in forecasting models. A natural extension would be to test the entire theoretical multipersistence object, computing richer multiparameter invariants or fibered barcode summaries for comparison against the sliced features inspected in this paper.  
\section{Conclusion}
The first significant result obtained is that compressed global family belief states improve forecasting, and that PCA is applicable within prediction markets. Hence, related prediction markets lie on a low-dimensional structure that linear compression can capture in part. 

The second result is that local belief-state topology is useful as a predictive factor. Moreover, multipersistence-inspired topology significantly benefited nonlinear models, allowing XGBoost and Random Forest to achieve exceptionally low log loss and Brier score metrics. 

Altogether, we conclude that application of persistent homology to correlated market neighborhoods, in combination with PCA-based global compression, preserves additional forecasting information that improves both linear and nonlinear model predictions. These findings suggest that prediction market belief states are partially low-dimensional, partially nonlinear, and that topology is highly valuable when utilized as a supplemental descriptive tool for local market geometry.

All code for this project can be accessed at the following repository: \newline https://github.com/ajerfy/Topology-of-Prediction-Market-Belief-Dynamics

\section{Acknowledgements}
I'd like to thank my mentor Elena Isasi Theus for all of her guidance throughout the process of writing this paper, as well as Peter May and the University of Chicago Mathematics REU for hosting fantastic lectures and providing an incredible research environment.

\end{document}